\documentclass[a4paper,UKenglish,cleveref,autoref,thm-restate]{lipics-v2021}

\pdfoutput=1 
\hideLIPIcs  

\title{Reconfiguring Subgraphs with Extra Resources} 

\author{Jason Fong}{Georgia Institute of Technology, United States}{jfong42@gatech.edu}{https://orcid.org/0009-0009-8379-9105}{}{}

\author{Jeffrey Kam}{University of Cambridge, United Kingdom}{hyk31@cantab.ac.uk}{https://orcid.org/0000-0002-4332-0026}{}

\author{Steven Wong}{University of Waterloo, Canada}{slnwong@uwaterloo.ca}{}{}{}

\authorrunning{J. Fong, J. Kam, and S. Wong} 

\Copyright{Jason Fong, Jeffrey Kam, and Steven Wong} 



\category{} 

\relatedversion{} 

\acknowledgements{Part of this work was presented in the Banff Reconfiguration Workshop (22w5090) and the authors would like to thank the participants for their suggestions. The authors would like to thank Naomi Nishimura and Benjamin Moore for the early discussions and feedback.}

\nolinenumbers 

\EventEditors{John Q. Open and Joan R. Access}
\EventNoEds{2}
\EventLongTitle{42nd Conference on Very Important Topics (CVIT 2016)}
\EventShortTitle{CVIT 2016}
\EventAcronym{CVIT}
\EventYear{2016}
\EventDate{December 24--27, 2016}
\EventLocation{Little Whinging, United Kingdom}
\EventLogo{}
\SeriesVolume{42}
\ArticleNo{23}

\usepackage{svg}
\usepackage{tikz}
\usetikzlibrary{patterns, decorations.pathmorphing}
\tikzset{
	baseline=(current bounding box.north),
	vertex/.style = {shape=circle, draw, minimum size=4mm, inner sep=0, font=\footnotesize},
	bvertex/.style = {shape=circle, draw, minimum size=1.5mm, inner sep=0, fill=black},
	edge/.style = {-, = latex'},
}
\tikzset{arc/.style = {->, = latex}}
\pgfdeclarelayer{bg}    
\pgfsetlayers{bg,main}  

\usepackage{multirow}   
\usepackage{makecell}   

\def\cP{\textsf{P}}
\def\cNP{\textsf{NP}}

\newcommand{\set}[1]{\{ #1 \}}

\newcommand{\floor}[1]{\lfloor #1 \rfloor}
\newcommand{\abs}[1]{| #1 |}

\begin{document}

\maketitle

\begin{abstract}
The \textsc{subgraph reconfiguration} problem asks whether one subgraph can be transformed into another via a sequence of local changes while maintaining a specified graph property.
In this work, we focus on the setting where the subgraph is specified by its set of edges.
Our contributions in this paper are twofold.
First, motivated by the contrast that path reconfiguration is $\textsf{NP}$-hard while tree reconfiguration is solvable in linear time, we prove two generalizations: (1) for any fixed $k$ at least one, reconfiguring connected graphs with pathwidth at most $k$ is $\textsf{NP}$-hard, and (2) for any fixed $k$ at least two, reconfiguring graphs with pathwidth at most $k$ is also $\textsf{NP}$-hard. 
En route to proving (2), we show a general hardness result that applies to a range of minor-closed graph classes, which we use to show planar graph reconfiguration is also $\textsf{NP}$-hard.
Second, given our negative results, we extend the problem to a resource-focused setting, asking how much additional buffer space is needed to turn a non-reconfigurable instance into a reconfigurable one. 
We show that $\Omega(n)$ extra buffer space is needed for planar graphs and graphs with bounded pathwidth and treewidth, while $O(1)$ extra buffer space is sufficient for cactus graphs in a restricted setting.
\end{abstract}

\section{Introduction}\label{sec:intro}

%

Combinatorial reconfiguration \cite{nishimura_18} is the study of transformations between configurations by making one local change at a time, providing a unifying framework for problems ranging from toy problems like the 15 puzzle to practical use-cases in power supply \cite{ito11} and quantum computing \cite{cimring_23, bansal_25}.
Our focus in this paper is on the \textsc{Subgraph reconfiguration} problem, first introduced by Hanaka et al. \cite{hanaka_20}.
In this setting, we are given an input graph $G$ and some graph property $\Pi$ (e.g. a graph is a tree). A configuration is a subgraph $H \subseteq G$ such that $H$ satisfies $\Pi$.
Specifically, we only consider the edge variant setting, where each configuration (a subgraph) is specified by its set of edges.
To align with existing literature, we imagine a placement of tokens on the edges of $G$ and the placement of tokens is a valid configuration if the token-induced subgraph, formed by edges with a token, satisfies $\Pi$.
The local change in our case, also known as a \emph{reconfiguration step}, is the movement of a token from one edge to an unoccupied edge, also known as the Token Jumping model.
We are most interested in the problem where given a source subgraph $E_s$ and target subgraph $E_t$ of the input graph $G$, both satisfying $\Pi$, whether there exists a sequence of local changes, also known as a \emph{reconfiguration sequence}, that transforms $E_s$ to $E_t$.
We call the problem the \textsc{$\Pi$ reconfiguration} problem.

Hanaka et al. \cite{hanaka_20} showed that \textsc{path reconfiguration} is \cNP-hard, whereas \textsc{tree reconfiguration} can be solved in linear time.
Motivated by this contrast, we investigate the extent to which structural graph properties influence reconfiguration complexity. 
In particular, we study the complexity of reconfiguration problems on graphs of bounded pathwidth and bounded treewidth, which are natural extensions of paths and trees, and we show that both problems are \textsf{NP}-hard.
Given the hardness results, we turn our focus to the resource-focused setting, where we are interested in knowing the impact of extra buffer space, i.e. a place to store tokens without invalidating the graph property, on instances of \textsc{subgraph reconfiguration}. 
In particular, we are interested in the minimum amount of buffer space required to turn a non-reconfigurable instance into a reconfigurable one.
We note that this model is similar to the Token Addition and Removal model \cite{ito11}, but with an additional constraint that we can only add as many tokens as we have removed thus far.

Prior work by Hanaka et al. \cite{hanaka_20} also considered other graph properties and their complexity under a combination of different rules (e.g. Token Sliding) and token-placement variants (e.g. placed on vertices). 
Demaine et al. \cite{demaine_19} considered a related path reconfiguration problem under the Token Sliding rule \cite{nishimura_18}, but their sliding rule concerns sliding the path as a whole, as opposed to sliding individual tokens. 
Other properties have been studied as well. For example, Bousquet et al. \cite{bousquet_20} considered the reconfiguration of spanning trees while preserving constraints on the number of leaves, while Eto et al. \cite{eto_22} studied the reconfiguration of $d$-regular induced subgraphs.
Extending upon the \textsc{Subgraph reconfiguration} problem, Moore et al. \cite{moore_18} also considered the reconfiguration of graph minors.

\begin{table}
\centering
\begin{tabular}{| c || c | c | c |}
    \hline
    \textbf{Property} & \textbf{Complexity} & \textbf{Constraints} & \textbf{Extra Buffer} \\
    \hline\hline
    Path  & \cNP-hard~\cite{hanaka_20} & Input is connected & $\Omega(n)$ [Theorem~\ref{thm:buffer_simple}] \\
    \hline
    Tree  & \cP~\cite{hanaka_20} & Input is connected & $0$ [Theorem~\ref{thm:buffer_simple}] \\
    \hline
    Cycle & \cP~\cite{hanaka_20} & Input is connected & $\infty$ [Theorem~\ref{thm:buffer_simple}] \\
    \hline
    \makecell{Connected\\ Pathwidth $\leq k$} & \cNP-hard~[Theorem~\ref{thm:pathwidth_k}] & Input is connected & $\Omega(n)$ [Theorem~\ref{thm:buffer_pw_tw}] \\
    \hline
    Treewidth $\leq k$ & \cNP-hard~[Theorem~\ref{thm:treewidth_reconf_np_hard}] & None & $\Omega(n)$ [Theorem~\ref{thm:buffer_pw_tw}] \\
    \hline
    Planar & \cNP-hard~[Theorem~\ref{thm:planar_reconf_np_hard}] & None & $\Omega(n)$ [Theorem~\ref{thm:buffer_pw_tw}] \\
    \hline
    \multirow{2}{*}{Cacti} & \multirow{2}{*}{Unknown}
        & None & $O(c_{\max})$ [Theorem~\ref{thm:buffer_cactus}] \\
    \cline{3-4}
        & & \makecell{Source and target are \\ triangular cacti} & $6$ [Theorem~\ref{thm:buffer_cactus}] \\
    \hline
\end{tabular}
\vspace{1mm}
\caption{
Summary of our main results.
Here, $n$ denotes the size of input, $c_{max}$ denotes the larger of the number of cycles in the source and target subgraphs, and an extra buffer of $\infty$ means no instance is reconfigurable even with extra buffer.
In the extra buffer setting, we assume the input graph is connected for the first four properties, as otherwise the problem trivializes when the source and target lie in different components.
}\label{table:main_result}
\end{table}

We want to emphasize here that the containment of graph properties does not imply containment of their reconfiguration complexities. 
If $\Pi \subseteq \Pi'$ and \textsc{$\Pi$-reconfiguration} is intractable (e.g. \cNP-hard), it does not imply that \textsc{$\Pi'$-reconfiguration} is also intractable.  
For example, \textsc{path reconfiguration} is \cNP-hard, but \textsc{tree reconfiguration} is solvable in linear time.
Conversely, if $\Pi \subseteq \Pi'$ and \textsc{$\Pi$-reconfiguration} is tractable (i.e. solvable in polynomial time), this does not imply \textsc{$\Pi'$-reconfiguration} is also tractable, as demonstrated by the hardness of reconfiguring graphs with treewidth at least $k$, for any fixed $k \geq 2$ (Theorem \ref{thm:treewidth_reconf_np_hard}), even though \textsc{tree reconfiguration} is tractable.
Restricting or relaxing the property can introduce obstructions that affect reconfiguration in either direction, so each property needs to be analyzed individually. 
This also underscores the importance of having general results that apply to a whole family of graph properties, such as our Theorem \ref{thm:minor_closed_hardness}, which applies to general minor-closed graph properties.

We organize the rest of the paper as follows.
Section \ref{sec:prelim} establishes the definitions and notations used throughout the paper.
Section \ref{sec:technical1} shows the hardness of reconfiguring graphs with bounded pathwidth.
Section \ref{sec:technical2} shows a general result for proving hardness of reconfiguring minor-closed properties that satisfy some special constraints, which allows us to show the hardness of reconfiguring both graphs with bounded treewidth and planar graphs.
Finally, in Section \ref{sec:technical3}, we introduce the problem of reconfiguring with extra buffer space and show both positive and negative results for a few graph classes. 
We conclude by discussing some future direction in Section \ref{sec:conclusion}.

\section{Preliminaries}\label{sec:prelim}

%

In this section, we introduce some notation and important background that will be used throughout the paper. We use an asterisk $(\ast)$ to denote results with deferred proofs in the Appendix. All graphs considered will be simple graphs.

First, we will formally introduce the subgraph reconfiguration problem.
Let $\Pi$ be a graph property (e.g. a graph is a path) and $G$ be an input graph.
A tuple $(G, E_s, E_t)$ is an instance of the \textsc{$\Pi$ reconfiguration} problem, where $E_s$ and $E_t$ are the edge subsets for the source and target respectively, and the subgraphs induced by $E_s$ and $E_t$ satisfy property $\Pi$.
A \emph{reconfiguration sequence} is a sequence $E_s = E_0, E_1, \ldots, E_\ell = E_t$ such that for $i \in \set{0, \ldots, \ell-1}$, $\abs{E_i} = \abs{E_{i+1}}$ and $\abs{E_{i} \Delta E_{i+1}} = 2$ where $\Delta$ denotes the symmetric difference, and for $j \in \set{0, \ldots, \ell}$, $E_{j}$ satisfies property $\Pi$.
We say that $(G, E_s, E_t)$ is \emph{reconfigurable} if there exists a reconfiguration sequence from $E_s$ to $E_t$. 
Otherwise, we say it is \emph{non-reconfigurable}.
The size of the instance $(G, E_s, E_t)$ is defined as the sum of $\abs{V(G)}$, $\abs{E(G)}$, $\abs{E_s}$, and $\abs{E_t}$.
We will assume throughout the paper that $\abs{E_s} = \abs{E_t}$ as otherwise reconfiguration is trivially impossible.
Also, in this paper we are only interested in the existence of such a reconfiguration sequence rather than finding a shortest one, which remains a direction for future work.

\begin{definition}[Tree Decomposition \cite{bodlaender_98}]\label{def:treewidth}
Let $G = (V,E)$ be a graph.
Let $T$ be a tree.
For each $t \in V(T)$, define a bag $B_t \subseteq V$.
Let $\mathcal{B} = \{B_t: t \in V(T)\}$ be a set of bags.
The pair $(T,\mathcal{B})$ is a \emph{tree decomposition} of $G$ if it satisfies the following:
\begin{enumerate}
    \item For $uv \in E$, $\set{u,v} \subseteq B_t$ for some $t \in V(T)$.
    \item Let $u \in V$. If $u \in B_{t_1}$ and $u \in B_{t_3}$ for some $t_1, t_3 \in V(T)$, then $u \in B_{t_2}$ for all $t_2$ on the unique path between $t_1$ and $t_3$ in $T$.
    \item $V = \bigcup_{t \in V(T)} B_t$.
\end{enumerate}
The \emph{width} of a tree decomposition is $\displaystyle \max_{t \in T} (\abs{B_t} - 1)$.
The \emph{treewidth} of $G$, denoted by $tw(G)$, is the minimum width of all tree decompositions of $G$. 
\end{definition}

\begin{definition}[Path Decomposition \cite{bodlaender_98}]\label{def:pathwidth}
    A path decomposition of $G = (V,E)$ is a tree decomposition where the underlying tree is a path. 
    It can be written as an ordered sequence of bags, $(B_i: 1 \leq i \leq m)$, where $m$ is the length of the path. 
    In particular, we can rewrite point $(2)$ as: For $u \in V$, if $u \in B_i$ and $u \in B_j$ for some $i,j \in \set{1, \ldots, n}$ where $i \leq j$, then $u \in B_k$ for all $i \leq k \leq j$.
    The \emph{pathwidth} of $G$, denoted by $pw(G)$, is the minimum width of all path decompositions of $G$. 
\end{definition}

We say $G$ has a \emph{$v$-path decomposition} if (1) there exists a path decomposition of $G$ with width $pw(G)$ and (2) $v \in V(G)$ is in the first bag of the decomposition.
For any fixed $k$, we say a graph $G$ has \emph{$k$-bounded treewidth} if it satisfies $tw(G) \leq k$, and \emph{connected $k$-bounded pathwidth} if $G$ is connected and satisfies $pw(G) \leq k$.

\begin{definition}[$k$-tree]
    A graph $G$ is called a \emph{$k$-tree} if it satisfies either (1) $G$ is a $(k+1)$-clique, or (2) there exists some vertex $v \in V(G)$ such that $G \setminus \set{v}$ is a $k$-tree, and $v$ has $k$ edges into $G \setminus \set{v}$.
\end{definition}
Note that $k$-trees are exactly the edge-maximal graphs with treewidth $k$, and any $k$-bounded treewidth graph is a subgraph of a $k$-tree \cite{nesetril_08}.

\begin{definition}[$1$-sum]
    Let $G_1$ and $G_2$ be distinct graphs with at least one vertex, and let $u \in V(G_1)$, $v \in V(G_2)$.
    Then the \emph{$1$-sum} of $G_1$ and $G_2$ at $u$ and $v$, denoted as $G_1 \oplus_{u, v} G_2$, is the graph obtained by identifying the vertex $u \in V(G_1)$ with vertex $v \in V(G_2)$. 
    We refer to this merged vertex as either $u$ or $v$ interchangeably.
    We write $G_1 \oplus_{v_1, u_1} G_2 \oplus_{v_2, u_2} \ldots \oplus_{v_{n-1}, u_{n-1}} G_n$ to refer to the iterative application of the $1$-sum operation from $G_1$ to $G_n$.
\end{definition}

\begin{definition}[$k$-subdivision]
    Let $G$ be a graph with $uv \in E(G)$, and let $k \geq 1$.
    A \emph{$k$-subdivision} of $uv$ results in a graph $G'$ where:
    \begin{enumerate}
        \item $V(G') = V(G) \cup \set{s_1, \ldots, s_k}$.
        \item $E(G') = (E(G) \setminus \set{uv}) \cup \set{us_1, \ldots, us_k, s_1v, \ldots, s_kv}$.
    \end{enumerate}
    We refer to these $s_i$'s as the intermediate vertices of $uv$ in $G'$.
    Furthermore, we write $E(uv, k, u)$ to denote the set of edges incident to $u$ introduced by the $k$-subdivision of $uv$, and similarly we write $E(uv, k, v)$ to denote the set of edges incident to $v$ introduced by the $k$-subdivision of $uv$.
\end{definition}

\begin{figure}[t]
\vspace{-0.25cm}
\centering
\begin{minipage}[b]{0.48\textwidth}
    \centering
    \subfloat[\centering $G_1$]{
    \begin{tikzpicture}[scale=0.7]
        \node[vertex] (v1) at (0,1) {$v_1$};
        \node[vertex] (v2) at (1,1) {$v_2$};
        \node[vertex] (v3) at (1,0) {$v_3$};
        \node[vertex] (v4) at (0,0) {$v_4$};
        \draw[edge] (v1) -- (v2);
        \draw[edge] (v2) -- (v3);
        \draw[edge] (v3) -- (v4);
        \draw[edge] (v4) -- (v1);
    \end{tikzpicture}
    }
    \quad
    \subfloat[\centering $G_2$]{
    \begin{tikzpicture}[scale=0.7]
        \node[vertex] (u1) at (0,0) {$u_1$};
        \node[vertex] (u2) at (1,-1) {$u_2$};
        \draw[edge] (u1) -- (u2);
    \end{tikzpicture}
    }
    \quad
    \subfloat[\centering]{
    \begin{tikzpicture}[scale=0.7]
        \node[vertex] (v1) at (0,1) {$v_1$};
        \node[vertex] (v2) at (1,1) {$v_2 / u_1$};
        \node[vertex] (v3) at (1,0) {$v_3$};
        \node[vertex] (v4) at (0,0) {$v_4$};
        \node[vertex] (u2) at (2,0) {$u_2$};
        \draw[edge] (v1) -- (v2);
        \draw[edge] (v2) -- (v3);
        \draw[edge] (v3) -- (v4);
        \draw[edge] (v4) -- (v1);
        \draw[edge] (v2) -- (u2);
    \end{tikzpicture}
    }
    \caption{Example of $G_1 \oplus_{v_2, u_1} G_2$}
    \label{fig:1_sum_example}
\end{minipage}
\hfill
\begin{minipage}[b]{0.48\textwidth}
    \centering
    \begin{tikzpicture}[scale=0.3]
        \node[vertex] (1) at (0,5) {$v_1$};
        \node[vertex] (2) at (5,5) {$v_2$};
        \node[vertex] (3) at (5,0) {$v_3$};
        \node[vertex] (4) at (0,0) {$v_4$};
        \draw[edge] (1) -- (2) node[midway, above] {a};
        \draw[edge] (2) -- (3) node[midway, right] {b};
        \draw[edge] (3) -- (4) node[midway, above] {c};
        \draw[edge] (4) -- (1) node[midway, left] {d};
    \end{tikzpicture}
    \quad
    \begin{tikzpicture}[scale=0.3]
        \node[vertex] (1) at (0,5) {$v_1$};
        \node[vertex] (2) at (5,5) {$v_2$};
    	\node[vertex] (3) at (5,0) {$v_3$};
    	\node[vertex] (4) at (0,0) {$v_4$};
    	
    	\draw[edge] (1) -- (2) node[midway, above] {a};
    	\draw[edge] (3) -- (4) node[midway, above] {c};
    	\draw[edge] (4) -- (1) node[midway, left] {d};
    
    	\node[vertex] (s1) at (3.5,2.5) {$s_{1}$};
        \node[vertex] (s2) at (5,2.5) {$s_{2}$};
        \node[vertex] (s3) at (6.5,2.5) {$s_{3}$};
    	
    	\draw[edge] (s1) -- (2);
    	\draw[edge] (s2) -- (2);
    	\draw[edge] (s3) -- (2);
    	\draw[edge] (s1) -- (3);
    	\draw[edge] (s2) -- (3);
    	\draw[edge] (s3) -- (3);
    \end{tikzpicture}
    \caption{Example of a $3$-subdivision of $v_2 v_3$}
    \label{fig:subdivision_example}
\end{minipage}
\vspace{-0.5cm}
\end{figure}

Finally, we introduce some terms and notations that we will use throughout the paper. We say a path $P = v_1 \ldots v_{k+1}$ has \emph{length} $k$. A tree $T$ with root vertex $r$ has \emph{depth} $k$ if the maximum distance from any leaf of $T$ to $r$ is $k$, where the distance is the length of the path between the two vertices. 
Furthermore, we say a tree is a \emph{perfect ternary tree} of depth $k$ if it is a ternary tree (i.e. each node has at most $3$ children) of depth $k$ with the maximum number of nodes. We write $[n]$ to denote the set $\set{1, \ldots, n}$ and $\mathbb{Z}_{\geq n}$ to denote $\set{n, n+1, \ldots}$.

\section{Graphs with Bounded Pathwidth}\label{sec:technical1}
Our goal in this section is to show the following result.

\begin{restatable}{theorem}{pathwidthbound}\label{thm:pathwidth_k}
 $(\ast)$ \textsc{connected $k$-bounded pathwidth reconfiguration} is $\cNP$-hard for any fixed $k \geq 1$.
\end{restatable}

To do so, we first show how we can prove this for a simpler case, where the graphs are connected with pathwidth at most one (or equivalently, caterpillar graphs \cite{kinnersley_94}), and then give a proof overview on how to extend this to any fixed pathwidth $k$ of at least two. 

First, we define the \textsc{Hamiltonian $v$-path} problem as a decision problem, where given $(G,v)$ it asks whether there exists a Hamiltonian path in $G$ with one of its endpoints as $v$, which we write compactly as a \emph{Hamiltonian $v$-path}. Note that the \textsc{Hamiltonian $v$-path} problem is $\cNP$-hard, which follows from a simple reduction from the standard \textsc{Hamiltonian path} problem \cite[p.~199]{garey_79}. 
The high-level overview of this proof is then a reduction from the \textsc{Hamiltonian $v$-path} problem to the \textsc{caterpillar reconfiguration} problem. To show this, we first introduce some definitions and a technical lemma.

\begin{definition}
    A \emph{caterpillar} is a tree in which the removal of all the leaf vertices of the tree results in a path.
    We call the resulting path the \emph{spine} of the caterpillar, and we call the removed vertices the \emph{legs} of the caterpillar.
    Moreover, we call the legs that are adjacent to the endpoints of the spine and the endpoints of the \emph{spine} themselves collectively the \emph{ends} of the caterpillar.
\end{definition}

\begin{definition}\label{def:spiked_graph}
    Let $G = (V, E)$ be a graph where $V = \set{v_1, \ldots, v_n}$. 
    We define the \emph{$k$-spiked graph of $G$} as $S$ where
    \[V(S) = V \cup \bigcup_{i=1}^{n} \set{s_{i,j}: 1 \leq j \leq k} \text{ and } E(S) = E \cup \bigcup_{i=1}^{n} \set{v_i s_{i,j}: 1 \leq j \leq k}.\]
    This is the graph $G$ with $k$ new vertices and edges for each existing vertex.
    We call each $v_i s_{i,j}$ a \emph{spike} of $G$ and an edge a \emph{$v$-spike} if $v$ is one of its endpoints.
\end{definition}

\begin{restatable}{lemma}{spikedhamiltonian}\label{lem:spike_hamiltonian}
    $(\ast)$ Let $P$ be a path of length at least two, and let $p_1, p_2$ be the endpoints of $P$.
    Let $S$ be the $k$-spiked graph of a graph $G$ and let $v$ be a specified vertex of $G$.
    Consider the graph $P \oplus_{p_2, v} S$ and let $H$ be a subgraph of this graph.
    
    If $H$ is a caterpillar that contains all of $P$ and uses at least one $s$-spike for all $s \in V(G)$, then $G$ contains a Hamiltonian path starting at $v$.
\end{restatable}

To show a reduction from the \textsc{Hamiltonian $v$-path problem}, we need to construct an instance $(G_{cater}, E_s, E_t)$ on which it is hard to reconfigure. Let $(G,v)$ be an instance of the \textsc{Hamiltonian $v$-path problem}.
Let $S$ be a $k$-spiked graph of $G$ (Definition \ref{def:spiked_graph}). We construct the input graph $G_{cater}$ by gluing $S$ with a ``fork''-like structure with two branches at the vertex $v$.
The source configuration $E_s$ will have all tokens on one branch, and the target configuration $E_t$ will have all tokens on the other branch. 
One cannot simply move tokens from one branch to another due to connectivity constraints and the ``fork''-like structure, which would make the token-induced subgraph not be a caterpillar. 
This forces the tokens to use edges in the graph $S$, and by Lemma \ref{lem:spike_hamiltonian}, we must use edges in a way that forms a Hamiltonian $v$-path, thus concluding our proof. 
We formalize the construction of $G_{cater}$ with the definition below.

\begin{definition}\label{def:aux-cater}
    Let $G = (V, E)$ be a graph where $V = \{v_1, \ldots, v_n\}$.
    Let $v$ be a vertex of $G$.
    We define the reduction graph $G_{cater}(v)$ as follows:
    \begin{enumerate}
        \item Let $P_1 = p_1^1 \ldots p_{1}^{2n}$ be a path of length $2n - 1$.
        \item Let $P_2 = p_2^1 \ldots p_{2}^{2n}$ be a path of length $2n - 1$.
        \item Let $Y$ be a claw (a star with $3$ edges) with vertices $\{a, b, c, r\}$, where $r$ is the center vertex, with an additional vertex $d$ and edge $c d$.
        \item Let $S$ be the $1$-spiked graph of $G$.
    \end{enumerate}
    Then $G_{cater}(v)$ is the graph formed by connecting $S, P_1, P_2$ by adding the edges $p_{1}^{2n} a$, $p_{2}^{2n} b$, and $d v$.
\end{definition}

We are now ready to prove the main result of this section.

\begin{theorem}\label{thm:caterpillar}
 The \textsc{caterpillar reconfiguration} problem is $\cNP$-hard.
\end{theorem}

\begin{proof}
    We will show this by a reduction from the \textsc{Hamiltonian $v$-path} problem.
    Let $(G, v)$ be an instance of the \textsc{Hamiltonian $v$-path} problem.
    Using $G$, we create the reduction graph $G_{cater}(v)$ from Definition \ref{def:aux-cater}.
    Let $E_s = E(P_1) \cup E(Y) \cup \{p_{1}^{2n} a, d v\}$ and $E_t = E(P_2) \cup E(Y) \cup \{p_{2}^{2n} b, d v\}$.
    Then we want to show that $G$ has a Hamiltonian path starting at $v$ if and only if the instance $(G_{cater}(v), E_s, E_t)$ is reconfigurable. For an illustrated example, see Figure \ref{fig:source_config} for the source configuration $E_s$ and Figure \ref{fig:target_config} for the target configuration $E_t$.
    
    Suppose that $G$ has a Hamiltonian path starting at $v$.
    Let this Hamiltonian path be $e_1 \ldots e_{n - 1}$, where $e_1$ is the edge starting at $v$.
    Note that the number of edges on this path is $n - 1$, and that the number of spikes in $G_{cater}(v)$ is $n$.
    Then one possible reconfiguration sequence from $E_s$ to $E_t$ in $G_{cater}(v)$ can be formed by performing the following:
    \begin{enumerate}
        \item First, move (token jump) $p_1^1 p_1^2$ to $e_1$, $p_1^2 p_1^3$ to $e_2$, $\ldots$, and $p_1^{n - 1} p_1^n$ to $e_{n - 1}$.
        \item Then, move $p_1^n p_1^{n + 1}$ to $v_1 s_1$, $p_1^{n + 1} p_1^{n + 2}$ to $v_2 s_2$, $\ldots$, $p_1^{2n - 1} p_1^{2n}$ to $v_{n} s_{n}$.
        \item Then, move $p_1^{2n}a$ to $p_1^{2n}b$.
        \item Finally, reverse all steps in $(2)$ and $(1)$, but now for the path $p_2^1 \ldots p_2^{2n}$.
    \end{enumerate}
                
    Note that each step in the reconfiguration sequence from (1) to (3) satisfies the caterpillar property because at each step, the removal of all the leaves in the subgraph formed by the tokens results in a path (we first move the spine, and then convert spine edges to legs).
    Since (4) just reverses all the steps in (2) and (1), each step also satisfies the caterpillar property.
    Hence, there exists a reconfiguration sequence from $E_s$ to $E_t$, so $(G_{cater}(v), E_s, E_t)$ is reconfigurable. 

    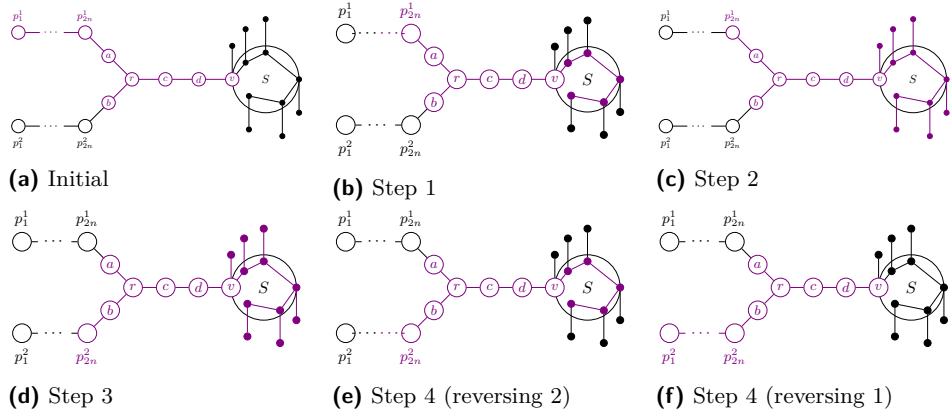
\begin{figure}[t]
\centering

\makebox[0.9\textwidth][c]{%

\begin{subfigure}[c]{0.28\textwidth}
\resizebox{\linewidth}{!}{%
\begin{tikzpicture}
    \tikzset{
        source/.style = {color=violet, prefix after command= {\pgfextra{\tikzset{every label/.style={violet}}}}},
    }
    \node[vertex, source] (r) at (0, 0) {$r$};
    \node[vertex, source] (a) at (-0.7, 0.7) {$a$};
    \node[vertex, source] (b) at (-0.7, -0.7) {$b$};
    \node[vertex, source] (c) at (1, 0) {$c$};
    \node[vertex, source] (d) at (2, 0) {$d$};
    
    \draw[edge, source] (r) -- (a);
    \draw[edge, source] (r) -- (b);
    \draw[edge, source] (r) -- (c);
    \draw[edge, source] (c) -- (d);
    
    \node[vertex, label={above:\small$p_1^1$}, source] (pt1) at (-3.4, 1.4) {};
    \node[source] (pt2) at (-2.4, 1.4) {$\cdots$};
    \node[vertex, label={above:\small$p_{2n}^1$}, source] (pt3) at (-1.4, 1.4) {};
    
    \draw[edge, source] (pt1) -- (pt2);
    \draw[edge, source] (pt2) -- (pt3);
    
    \node[vertex, label={below:\small$p_1^2$}] (pb1) at (-3.4, -1.4) {};
    \node[] (pb2) at (-2.4, -1.4) {$\cdots$};
    \node[vertex, label={below:\small$p_{2n}^2$}] (pb3) at (-1.4, -1.4) {};
    
    \draw[edge] (pb1) -- (pb2);
    \draw[edge] (pb2) -- (pb3);
    
    \node[vertex, source, fill=white] (v) at (3, 0) {$v$};
    \begin{pgfonlayer}{bg}
        \draw[] (4,0) circle (1);
    \end{pgfonlayer}
    \node[] at (4, 0) {$S$};
    
    \node[bvertex] (p1) at (3.4, 0.5) {};
    \node[bvertex] (p2) at (4, 0.8) {};
    \node[bvertex] (p3) at (5, 0) {};
    \node[bvertex] (p4) at (4.5, -0.7) {};
    \node[bvertex] (p5) at (3.5, -0.5) {};
    
    \node[bvertex] (sv) at (3, 1) {};
    \node[bvertex] (s1) at (3.4, 1.5) {};
    \node[bvertex] (s2) at (4, 1.8) {};
    \node[bvertex] (s3) at (5, -1) {};
    \node[bvertex] (s4) at (4.5, -1.7) {};
    \node[bvertex] (s5) at (3.5, -1.5) {};
    
    \draw[edge] (v) -- (p1);
    \draw[edge] (p1) -- (p2);
    \draw[edge] (p2) -- (p3);
    \draw[edge] (p3) -- (p4);
    \draw[edge] (p4) -- (p5);
    
    \draw[edge] (v) -- (sv);
    \draw[edge] (p1) -- (s1);
    \draw[edge] (p2) -- (s2);
    \draw[edge] (p3) -- (s3);
    \draw[edge] (p4) -- (s4);
    \draw[edge] (p5) -- (s5);
    
    \draw[edge, source] (pt3) -- (a);
    \draw[edge] (pb3) -- (b);
    \draw[edge, source] (d) -- (v);
\end{tikzpicture}
}
\caption{Initial}\label{fig:source_config}
\end{subfigure}

\hfill

\begin{subfigure}[c]{0.28\textwidth}
\resizebox{\linewidth}{!}{%
\begin{tikzpicture}[scale=0.7]
    \tikzset{
        source/.style = {color=violet, prefix after command= {\pgfextra{\tikzset{every label/.style={violet}}}}},
    }
    \node[vertex, source] (r) at (0, 0) {$r$};
    \node[vertex, source] (a) at (-0.7, 0.7) {$a$};
    \node[vertex, source] (b) at (-0.7, -0.7) {$b$};
    \node[vertex, source] (c) at (1, 0) {$c$};
    \node[vertex, source] (d) at (2, 0) {$d$};
    
    \draw[edge, source] (r) -- (a);
    \draw[edge, source] (r) -- (b);
    \draw[edge, source] (r) -- (c);
    \draw[edge, source] (c) -- (d);
    
    \node[vertex, label={above:\small$p_1^1$}] (pt1) at (-3.4, 1.4) {};
    \node[] (pt2) at (-2.4, 1.4) {$\cdots\textcolor{violet}{\cdots}$};
    \node[vertex, label={above:\small$p_{2n}^1$}, source] (pt3) at (-1.4, 1.4) {};
    
    \draw[edge] (pt1) -- (pt2);
    \draw[edge, source] (pt2) -- (pt3);
    
    \node[vertex, label={below:\small$p_1^2$}] (pb1) at (-3.4, -1.4) {};
    \node[] (pb2) at (-2.4, -1.4) {$\cdots$};
    \node[vertex, label={below:\small$p_{2n}^2$}] (pb3) at (-1.4, -1.4) {};
    
    \draw[edge] (pb1) -- (pb2);
    \draw[edge] (pb2) -- (pb3);
    
    \node[vertex, source, fill=white] (v) at (3, 0) {$v$};
    \begin{pgfonlayer}{bg}
        \draw[] (4,0) circle (1);
    \end{pgfonlayer}
    \node[] at (4, 0) {$S$};
    
    \node[bvertex, source] (p1) at (3.4, 0.5) {};
    \node[bvertex, source] (p2) at (4, 0.8) {};
    \node[bvertex, source] (p3) at (5, 0) {};
    \node[bvertex, source] (p4) at (4.5, -0.7) {};
    \node[bvertex, source] (p5) at (3.5, -0.5) {};
    
    \node[bvertex] (sv) at (3, 1) {};
    \node[bvertex] (s1) at (3.4, 1.5) {};
    \node[bvertex] (s2) at (4, 1.8) {};
    \node[bvertex] (s3) at (5, -1) {};
    \node[bvertex] (s4) at (4.5, -1.7) {};
    \node[bvertex] (s5) at (3.5, -1.5) {};
    
    \draw[edge, source] (v) -- (p1);
    \draw[edge, source] (p1) -- (p2);
    \draw[edge, source] (p2) -- (p3);
    \draw[edge, source] (p3) -- (p4);
    \draw[edge, source] (p4) -- (p5);
    
    \draw[edge] (v) -- (sv);
    \draw[edge] (p1) -- (s1);
    \draw[edge] (p2) -- (s2);
    \draw[edge] (p3) -- (s3);
    \draw[edge] (p4) -- (s4);
    \draw[edge] (p5) -- (s5);
    
    \draw[edge, source] (pt3) -- (a);
    \draw[edge] (pb3) -- (b);
    \draw[edge, source] (d) -- (v);
\end{tikzpicture}
}
\caption{Step $1$}
\end{subfigure}

\hfill

\begin{subfigure}[c]{0.28\textwidth}
\resizebox{\linewidth}{!}{%
\begin{tikzpicture}
    \tikzset{
        source/.style = {color=violet, prefix after command= {\pgfextra{\tikzset{every label/.style={violet}}}}},
    }
    \node[vertex, source] (r) at (0, 0) {$r$};
    \node[vertex, source] (a) at (-0.7, 0.7) {$a$};
    \node[vertex, source] (b) at (-0.7, -0.7) {$b$};
    \node[vertex, source] (c) at (1, 0) {$c$};
    \node[vertex, source] (d) at (2, 0) {$d$};
    
    \draw[edge, source] (r) -- (a);
    \draw[edge, source] (r) -- (b);
    \draw[edge, source] (r) -- (c);
    \draw[edge, source] (c) -- (d);
    
    \node[vertex, label={above:\small$p_1^1$}] (pt1) at (-3.4, 1.4) {};
    \node[] (pt2) at (-2.4, 1.4) {$\cdots$};
    \node[vertex, source, label={above:\small$p_{2n}^1$}] (pt3) at (-1.4, 1.4) {};
    
    \draw[edge] (pt1) -- (pt2);
    \draw[edge] (pt2) -- (pt3);
    
    \node[vertex, label={below:\small$p_1^2$}] (pb1) at (-3.4, -1.4) {};
    \node[] (pb2) at (-2.4, -1.4) {$\cdots$};
    \node[vertex, label={below:\small$p_{2n}^2$}] (pb3) at (-1.4, -1.4) {};
    
    \draw[edge] (pb1) -- (pb2);
    \draw[edge] (pb2) -- (pb3);
    
    \node[vertex, source, fill=white] (v) at (3, 0) {$v$};
    \begin{pgfonlayer}{bg}
        \draw[] (4,0) circle (1);
    \end{pgfonlayer}
    \node[] at (4, 0) {$S$};
    
    \node[bvertex, source] (p1) at (3.4, 0.5) {};
    \node[bvertex, source] (p2) at (4, 0.8) {};
    \node[bvertex, source] (p3) at (5, 0) {};
    \node[bvertex, source] (p4) at (4.5, -0.7) {};
    \node[bvertex, source] (p5) at (3.5, -0.5) {};
    
    \node[bvertex, source] (sv) at (3, 1) {};
    \node[bvertex, source] (s1) at (3.4, 1.5) {};
    \node[bvertex, source] (s2) at (4, 1.8) {};
    \node[bvertex, source] (s3) at (5, -1) {};
    \node[bvertex, source] (s4) at (4.5, -1.7) {};
    \node[bvertex, source] (s5) at (3.5, -1.5) {};
    
    \draw[edge, source] (v) -- (p1);
    \draw[edge, source] (p1) -- (p2);
    \draw[edge, source] (p2) -- (p3);
    \draw[edge, source] (p3) -- (p4);
    \draw[edge, source] (p4) -- (p5);
    
    \draw[edge, source] (v) -- (sv);
    \draw[edge, source] (p1) -- (s1);
    \draw[edge, source] (p2) -- (s2);
    \draw[edge, source] (p3) -- (s3);
    \draw[edge, source] (p4) -- (s4);
    \draw[edge, source] (p5) -- (s5);
    
    \draw[edge, source] (pt3) -- (a);
    \draw[edge] (pb3) -- (b);
    \draw[edge, source] (d) -- (v);
\end{tikzpicture}
}
\caption{Step $2$}
\end{subfigure}

}

\makebox[0.9\textwidth][c]{%

\begin{subfigure}[c]{0.28\textwidth}
\resizebox{\linewidth}{!}{%
\begin{tikzpicture}[scale=0.7]
    \tikzset{
        source/.style = {color=violet, prefix after command= {\pgfextra{\tikzset{every label/.style={violet}}}}},
    }
    \node[vertex, source] (r) at (0, 0) {$r$};
    \node[vertex, source] (a) at (-0.7, 0.7) {$a$};
    \node[vertex, source] (b) at (-0.7, -0.7) {$b$};
    \node[vertex, source] (c) at (1, 0) {$c$};
    \node[vertex, source] (d) at (2, 0) {$d$};
    
    \draw[edge, source] (r) -- (a);
    \draw[edge, source] (r) -- (b);
    \draw[edge, source] (r) -- (c);
    \draw[edge, source] (c) -- (d);
    
    \node[vertex, label={above:\small$p_1^1$}] (pt1) at (-3.4, 1.4) {};
    \node[] (pt2) at (-2.4, 1.4) {$\cdots$};
    \node[vertex, label={above:\small$p_{2n}^1$}] (pt3) at (-1.4, 1.4) {};
    
    \draw[edge] (pt1) -- (pt2);
    \draw[edge] (pt2) -- (pt3);
    
    \node[vertex, label={below:\small$p_1^2$}] (pb1) at (-3.4, -1.4) {};
    \node[] (pb2) at (-2.4, -1.4) {$\cdots$};
    \node[vertex, source, label={below:\small$p_{2n}^2$}] (pb3) at (-1.4, -1.4) {};
    
    \draw[edge] (pb1) -- (pb2);
    \draw[edge] (pb2) -- (pb3);
    
    \node[vertex, source, fill=white] (v) at (3, 0) {$v$};
    \begin{pgfonlayer}{bg}
        \draw[] (4,0) circle (1);
    \end{pgfonlayer}
    \node[] at (4, 0) {$S$};
    
    \node[bvertex, source] (p1) at (3.4, 0.5) {};
    \node[bvertex, source] (p2) at (4, 0.8) {};
    \node[bvertex, source] (p3) at (5, 0) {};
    \node[bvertex, source] (p4) at (4.5, -0.7) {};
    \node[bvertex, source] (p5) at (3.5, -0.5) {};
    
    \node[bvertex, source] (sv) at (3, 1) {};
    \node[bvertex, source] (s1) at (3.4, 1.5) {};
    \node[bvertex, source] (s2) at (4, 1.8) {};
    \node[bvertex, source] (s3) at (5, -1) {};
    \node[bvertex, source] (s4) at (4.5, -1.7) {};
    \node[bvertex, source] (s5) at (3.5, -1.5) {};
    
    \draw[edge, source] (v) -- (p1);
    \draw[edge, source] (p1) -- (p2);
    \draw[edge, source] (p2) -- (p3);
    \draw[edge, source] (p3) -- (p4);
    \draw[edge, source] (p4) -- (p5);
    
    \draw[edge, source] (v) -- (sv);
    \draw[edge, source] (p1) -- (s1);
    \draw[edge, source] (p2) -- (s2);
    \draw[edge, source] (p3) -- (s3);
    \draw[edge, source] (p4) -- (s4);
    \draw[edge, source] (p5) -- (s5);
    
    \draw[edge] (pt3) -- (a);
    \draw[edge, source] (pb3) -- (b);
    \draw[edge, source] (d) -- (v);
\end{tikzpicture}
}
\caption{Step $3$}
\end{subfigure}

\hfill

\begin{subfigure}[c]{0.28\textwidth}
\resizebox{\linewidth}{!}{%
\begin{tikzpicture}[scale=0.7]
    \tikzset{
        source/.style = {color=violet, prefix after command= {\pgfextra{\tikzset{every label/.style={violet}}}}},
    }
    \node[vertex, source] (r) at (0, 0) {$r$};
    \node[vertex, source] (a) at (-0.7, 0.7) {$a$};
    \node[vertex, source] (b) at (-0.7, -0.7) {$b$};
    \node[vertex, source] (c) at (1, 0) {$c$};
    \node[vertex, source] (d) at (2, 0) {$d$};
    
    \draw[edge, source] (r) -- (a);
    \draw[edge, source] (r) -- (b);
    \draw[edge, source] (r) -- (c);
    \draw[edge, source] (c) -- (d);
    
    \node[vertex, label={above:\small$p_1^1$}] (pt1) at (-3.4, 1.4) {};
    \node[] (pt2) at (-2.4, 1.4) {$\cdots$};
    \node[vertex, label={above:\small$p_{2n}^1$}] (pt3) at (-1.4, 1.4) {};
    
    \draw[edge] (pt1) -- (pt2);
    \draw[edge] (pt2) -- (pt3);
    
    \node[vertex, label={below:\small$p_1^2$}] (pb1) at (-3.4, -1.4) {};
    \node[] (pb2) at (-2.4, -1.4) {$\cdots\textcolor{violet}{\cdots}$};
    \node[vertex, source, label={below:\small$p_{2n}^2$}] (pb3) at (-1.4, -1.4) {};
    
    \draw[edge] (pb1) -- (pb2);
    \draw[edge, source] (pb2) -- (pb3);
    
    \node[vertex, source, fill=white] (v) at (3, 0) {$v$};
    \begin{pgfonlayer}{bg}
        \draw[] (4,0) circle (1);
    \end{pgfonlayer}
    \node[] at (4, 0) {$S$};
    
    \node[bvertex, source] (p1) at (3.4, 0.5) {};
    \node[bvertex, source] (p2) at (4, 0.8) {};
    \node[bvertex, source] (p3) at (5, 0) {};
    \node[bvertex, source] (p4) at (4.5, -0.7) {};
    \node[bvertex, source] (p5) at (3.5, -0.5) {};
    
    \node[bvertex] (sv) at (3, 1) {};
    \node[bvertex] (s1) at (3.4, 1.5) {};
    \node[bvertex] (s2) at (4, 1.8) {};
    \node[bvertex] (s3) at (5, -1) {};
    \node[bvertex] (s4) at (4.5, -1.7) {};
    \node[bvertex] (s5) at (3.5, -1.5) {};
    
    \draw[edge, source] (v) -- (p1);
    \draw[edge, source] (p1) -- (p2);
    \draw[edge, source] (p2) -- (p3);
    \draw[edge, source] (p3) -- (p4);
    \draw[edge, source] (p4) -- (p5);
    
    \draw[edge] (v) -- (sv);
    \draw[edge] (p1) -- (s1);
    \draw[edge] (p2) -- (s2);
    \draw[edge] (p3) -- (s3);
    \draw[edge] (p4) -- (s4);
    \draw[edge] (p5) -- (s5);
    
    \draw[edge] (pt3) -- (a);
    \draw[edge, source] (pb3) -- (b);
    \draw[edge, source] (d) -- (v);
\end{tikzpicture}
}
\caption{Step $4$ (reversing $2$)}
\end{subfigure}

\hfill

\begin{subfigure}[c]{0.28\textwidth}
\resizebox{\linewidth}{!}{%
\begin{tikzpicture}[scale=0.7]
    \tikzset{
        source/.style = {color=violet, prefix after command= {\pgfextra{\tikzset{every label/.style={violet}}}}},
    }
    \node[vertex, source] (r) at (0, 0) {$r$};
    \node[vertex, source] (a) at (-0.7, 0.7) {$a$};
    \node[vertex, source] (b) at (-0.7, -0.7) {$b$};
    \node[vertex, source] (c) at (1, 0) {$c$};
    \node[vertex, source] (d) at (2, 0) {$d$};
    
    \draw[edge, source] (r) -- (a);
    \draw[edge, source] (r) -- (b);
    \draw[edge, source] (r) -- (c);
    \draw[edge, source] (c) -- (d);
    
    \node[vertex, label={above:\small$p_1^1$}] (pt1) at (-3.4, 1.4) {};
    \node[] (pt2) at (-2.4, 1.4) {$\cdots$};
    \node[vertex, label={above:\small$p_{2n}^1$}] (pt3) at (-1.4, 1.4) {};
    
    \draw[edge] (pt1) -- (pt2);
    \draw[edge] (pt2) -- (pt3);
    
    \node[vertex, label={below:\small$p_1^2$}, source] (pb1) at (-3.4, -1.4) {};
    \node[source] (pb2) at (-2.4, -1.4) {$\cdots$};
    \node[vertex, label={below:\small$p_{2n}^2$}, source] (pb3) at (-1.4, -1.4) {};
    
    \draw[edge, source] (pb1) -- (pb2);
    \draw[edge, source] (pb2) -- (pb3);
    
    \node[vertex, source, fill=white] (v) at (3, 0) {$v$};
    \begin{pgfonlayer}{bg}
        \draw[] (4,0) circle (1);
    \end{pgfonlayer}
    \node[] at (4, 0) {$S$};
    
    \node[bvertex] (p1) at (3.4, 0.5) {};
    \node[bvertex] (p2) at (4, 0.8) {};
    \node[bvertex] (p3) at (5, 0) {};
    \node[bvertex] (p4) at (4.5, -0.7) {};
    \node[bvertex] (p5) at (3.5, -0.5) {};
    
    \node[bvertex] (sv) at (3, 1) {};
    \node[bvertex] (s1) at (3.4, 1.5) {};
    \node[bvertex] (s2) at (4, 1.8) {};
    \node[bvertex] (s3) at (5, -1) {};
    \node[bvertex] (s4) at (4.5, -1.7) {};
    \node[bvertex] (s5) at (3.5, -1.5) {};
    
    \draw[edge] (v) -- (p1);
    \draw[edge] (p1) -- (p2);
    \draw[edge] (p2) -- (p3);
    \draw[edge] (p3) -- (p4);
    \draw[edge] (p4) -- (p5);
    
    \draw[edge] (v) -- (sv);
    \draw[edge] (p1) -- (s1);
    \draw[edge] (p2) -- (s2);
    \draw[edge] (p3) -- (s3);
    \draw[edge] (p4) -- (s4);
    \draw[edge] (p5) -- (s5);
    
    \draw[edge] (pt3) -- (a);
    \draw[edge, source] (pb3) -- (b);
    \draw[edge, source] (d) -- (v);
\end{tikzpicture}
}
\caption{Step $4$ (reversing $1$)}\label{fig:target_config}
\end{subfigure}

}

\caption{The reconfiguration sequence for $(G'_{cater}(v), E_s, E_t)$ in Theorem \ref{thm:caterpillar}.}
\label{fig:caterpillar_reconf_hamiltonian}
\vspace{-0.5cm}
\end{figure}

    Now on the other hand, suppose that $(G_{cater}(v), E_s, E_t)$ is reconfigurable.
    This means that there exists a step in the reconfiguration sequence in which a token is moved to $p_2^{2n} b$, as this edge is part of the target solution.
    Consider the placement of tokens at this step.
    
    Note that until (at least) the edges $p_1^1 p_1^2, \ldots, p_1^{2n - 1} p_1^{2n}$ are moved, we cannot move any token to $p_2^{2n} b$, as otherwise the subgraph induced by the edges would not satisfy the caterpillar property. This is because either the token-induced graph becomes disconnected or has pathwidth at least two. 
    This is a total of $2n - 1$ tokens that need to be moved, and the only place for these tokens to move to is in $S$, as the subgraph at each step must be connected.

    Also, note that a caterpillar does not have cycles, as it is a tree.
    So the maximum number of tokens that can be placed in $S$ while maintaining caterpillar is $(n - 1) + n = 2n - 1$, which is when the tokens in $G$ form a spanning tree of $G$, since a spanning tree of a graph with $n$ vertices has $n - 1$ edges, and tokens are on all of the spikes of $S$ (there are $n$ spikes).
    This means that the $2n - 1$ tokens must have been all moved to $S$ during the reconfiguration.
    Now, restrict our view of the subgraph formed by the tokens to just those in $S$ along with the edges $c d, d v$ and call this $H$.
    As $cdv$ is a path of length $2$, $S$ is the $1$-spiked graph of $G$, and $H$ is a caterpillar, by Lemma \ref{lem:spike_hamiltonian}, we have that $G$ contains a Hamiltonian path starting at $v$.
    This reduction takes polynomial time since our construction depends linearly on the input size.
\end{proof}

Now, we will discuss how to extend the above result from pathwidth one into any pathwidth $k \geq 2$. More formally, we want to show the following theorem.

\pathwidthbound* \label{thm:pathwidth_k_proof_overview}

The proof structure remains similar. First, we start with an instance of the Hamiltonian $v$-path problem $(G,v)$. 
But unlike Definition \ref{def:aux-cater}, the construction here is a bit more involved. 
Roughly speaking, we construct the reduction graph $G_{path}(v)$ by performing the following:

\begin{enumerate}
    \item First, start with a perfect ternary tree $T$ of depth $k$. 
    Let $r$ be the root node, and $rv_1$, $rv_2$, and $rv_3$ be the three edges adjacent to $r$.
    \item For each leaf node $w$, we take the $1$-sum at $w$ with an endpoint of a path of length $f(k,n)$, where $n = \abs{G}$ and $f(k,n)$ is a function of order $O(3^{2k} n^4)$. We say a node $w$ belongs to branch $i$ if the shortest $(r,w)$-path contains the edge $rv_i$.
    \item Then, for each resulting leaf node $w$ that belongs to branch $3$, we make a $g(k,n)$-spiked graph with a copy of $G$, and take the $1$-sum at $w$ and $v'$ corresponding to $v$ in the copy, where $g(k,n)$ is a function of order $O(3^k n^3)$.
    
\end{enumerate}

With the reduction graph, we then need to define the source and target configuration $E_s$ and $E_t$ on $G_{path}(v)$ such that $(G_{path}(v), E_s, E_t)$ is reconfigurable if and only if $G$ has a Hamiltonian $v$-path. 
First, group the leaf nodes into three lists, where list $i$ belongs to branch $i$ for $i \in [3]$, and furthermore, assume an ordering of the nodes within the lists.
Then, we define $E_s$ and $E_t$ as follows: $E_s$ consists of all edges in $T$ with all the appended paths of length $f(k,n)$, but excludes the edges in the path $1$-summed with the first node in list $2$, which we call $P_{S \setminus T}$. Similarly, $E_t$ consists of all edges in $T$ with all the appended paths of length $f(k,n)$, but excludes the edges in the path $1$-summed with the first node in list $1$, which we call $P_{T \setminus S}$.

We claim that $G$ has a Hamiltonian $v$-path if and only if $(G_{path}(v), E_s, E_t)$ is reconfigurable. The bidirectional proof is structurally similar to that of Theorem \ref{thm:caterpillar}. For the forward direction, we show an explicit reconfiguration sequence utilizing the Hamiltonian $v$-path. For the backwards direction, we show that to move a token from $P_{S \setminus T}$ to $P_{T \setminus S}$, we must first move $O(f(k,n))$ tokens to edges on the spiked graph, which by a counting argument, means we have a structure that satisfies Lemma \ref{lem:spike_hamiltonian}, so we have a Hamiltonian $v$-path.

\section{Graphs with Bounded Treewidth and Planar Graphs}\label{sec:technical2}

%
Our main result in this section is a general theorem for showing hardness results for \emph{special} minor-closed graph classes, which contains the class of bounded treewidth graphs and planar graphs.
Specifically, we say a minor-closed graph class $\Pi$ is \emph{special} if it: (1) is closed under $k$-subdivision, (2) is closed under $1$-sum, (3) contains a graph with at least one edge. This section is organized into three parts. In Section \ref{subsec:reconf_special_minor}, we will prove the general theorem on the hardness of reconfiguring \emph{special} minor-closed graph classes, and then in Section \ref{subsec:reconf_treewidth} and  \ref{subsec:reconf_planar}, we show how to apply the theorem to get the desired hardness result for reconfiguring graphs with bounded treewidth and planar graphs, respectively.

\subsection{Special Minor-closed Graph Class}\label{subsec:reconf_special_minor}
First, let us introduce some common notation needed for the following sections. 
Let $\Pi$ be a graph property. The \textsc{Decide($\Pi$)} problem is the decision problem that asks, given a connected graph $G$ and an integer $\ell \leq \abs{E(G)}$, whether $G$ has a subgraph $H$ such that $\abs{E(H)} \geq \ell$ and $H$ satisfies $\Pi$. We write the input to the \textsc{Decide($\Pi$)} as $(G,\ell)$.
When the context is clear, we sometimes refer to the graph class $\Pi$ as a graph property, which simply means whether a graph belongs to $\Pi$.
Our goal in this section is to relate the hardness of \textsc{Decide($\Pi$)} to the hardness of \textsc{$\Pi$ reconfiguration}, as shown in the following theorem.

\begin{restatable}{theorem}{minorclosedhardness}\label{thm:minor_closed_hardness}
    Let $\Pi$ be a special minor-closed graph class.
    If \textsc{Decide($\Pi$)} is $\cNP$-hard, then \textsc{$\Pi$ reconfiguration} is $\cNP$-hard.
\end{restatable}



\begin{proof}
    Suppose \textsc{Decide($\Pi$)} is $\cNP$-hard.
    Let $(G, \ell)$ be an input for the \textsc{Decide($\Pi$)} problem.
    Let $n = \abs{V(G)}$ and $m = \abs{E(G)}$.
    Our approach here is to construct an instance $\mathcal{I}$ of the \textsc{$\Pi$ reconfiguration} problem such that $G$ is a yes-instance for \textsc{Decide($\Pi$)} if and only if $\mathcal{I}$ is reconfigurable.
    
    By the Robertson-Seymour theorem \cite{robertson_04}, there exists a finite set $\mathcal{F}$ of forbidden minors that describes $\Pi$. More specifically, any graph $G \in \Pi$ must not contain any $F \in \mathcal{F}$ as minors.
    So, from the set of forbidden minors $\mathcal{F}$, we can choose an $H \in \mathcal{F}$ with the least amount of edges (arbitrarily when there are multiple). 
    
    We claim that $E(H) \geq 2$.
    Suppose for the contrary, let $H$ be a forbidden minor of $\Pi$ with only $0$ or $1$ edge. 
    Since $\Pi$ is \emph{special}, it contains a graph $G$ with at least one edge. In particular, since $\Pi$ is minor-closed, the single-edge graph is in $\Pi$. 
    Also, since it is closed under $1$-sum, we can construct an arbitrarily long path $P$ by $1$-summing the single edge.
    In particular, we have $\abs{V(P)} \geq \abs{V(H)}$, where $E(P) \geq 1$, so it contains $H$ as a minor since we can just remove the appropriate edges and vertices, contradicting that $H$ is a forbidden minor.
    Hence, $E(H) \geq 2$.
    
    Let $uv$ and $st$ be two edges of $H$ and let $H'$ be formed by a $2(\ell+1)$-subdivision on all the edges except $uv$, and an $(\ell+1)$-subdivision of edge $uv$.
    Note that $H' \not\in \Pi$, since $H'$ contains $H$ as a minor.
    
    Define $G_A = H' \cup G$ as the auxiliary graph.
    Let $E_s = E(H') \setminus (E(uv,\ell+1,u) \cup E(uv,\ell+1,v))$ and $E_t = E(H') \setminus E(st,2(\ell+1),s)$.
    Let $(G_A, E_s, E_t)$ be an instance to the \textsc{$\Pi$ reconfiguration} problem. It is easy to see that the instance can be constructed in time polynomial to the input size.
    
    Note that $E_s$ satisfies $\Pi$ since $E_s$ can be constructed from $H \setminus \set{uv}$ by taking a $2(\ell+1)$-subdivision on all its edges. 
    This is because $\abs{E(H) \setminus \set{uv}} < \abs{E(H)}$ and $H$ is a forbidden minor with minimum number of edges, so $H \setminus \set{uv}$ satisfies $\Pi$, and that $\Pi$ is closed under taking $2(\ell+1)$-subdivision.
    For $E_t$, we know that $H \setminus \set{st}$ satisfies $\Pi$, and that any single edge graph satisfies $\Pi$, since any single edge graph cannot be a forbidden minor as $E(H) \geq 2$. So, we can construct $E_t$ by taking a $2(\ell + 1)$-subdivision on all edges of $H \setminus \set{st}$ except $uv$, a $(\ell+1)$-subdivision on $uv$, and then $1$-summing $s$ with $k$ single edge graphs. Since $\Pi$ is closed under taking $k$-subdivision and $1$-sum, we have $E_t \in \Pi$ as well.
    Now, we are ready to show that $(G_{A}, E_s, E_t)$ is reconfigurable if and only if \textsc{Decide($\Pi$)} is a yes-instance.
    
    Suppose $(G_{A}, E_s, E_t)$ is reconfigurable.
    Since $E_t$ includes the subdivided edges of $uv$, at some point in the reconfiguration sequence, there must be more than $\ell+1$ tokens in $E(uv,\ell+1,u) \cup E(uv,\ell+1,v)$.
    Consider the step $t$ of the reconfiguration sequence, right after we place the $\ell+2$ tokens onto $E(uv,\ell+1,u) \cup E(uv,\ell+1,v)$.
    We claim that if there are $\ell+2$ tokens in $E(uv,\ell+1,u) \cup E(uv,\ell+1,v)$, there must be some edge $ab \in E(H) \setminus \set{uv}$ such that at least $2(\ell+1)$ tokens in $E(ab,2(\ell+1),a) \cup E(ab,2(\ell+1),b)$ must be moved elsewhere.
    Suppose not, for all $ab \in E(H) \setminus \set{uv}$, only at most $2\ell + 1$ tokens are moved elsewhere. 
    Choose an edge $ab \in E(H) \setminus \set{uv}$.
    Since $H'$ is constructed by taking $2(\ell+1)$-subdivision of all edges except $uv$, if we remove only at most $2\ell+1$ tokens on $E(ab,2(\ell+1),a) \cup E(ab,2(\ell+1),b)$, there will be an intermediate vertex $s_{ab}$ of $ab$ in $H'$ where there is a token placed on both $as_{ab}$ and $s_{ab}b$. 
    Furthermore, by construction of $H'$, if there are $\ell+2$ tokens in $E(uv,\ell+1,u) \cup E(uv,\ell+1,v)$, there must also be an intermediate vertex $s_{uv}$ where there is a token placed on both $us_{uv}$ and $s_{uv}v$.
    However, this means the subgraph induced by the tokens in $H'$ contains $H$ as a minor, since for any edge $xy \in H$, there is a token placed on both $xs_{xy}$ and $s_{xy}y$, for some intermediate vertex $s_{xy}$ of $xy$ in $H'$.
    This means the subgraph induced by the tokens does not satisfy $\Pi$, which is a contradiction.

    In particular, at step $t$, this means at least $2(\ell+1)$ tokens in $H' \setminus (E(uv,\ell+1,u) \cup E(uv,\ell+1,v))$ must be moved away into $E(uv,\ell+1,u) \cup E(uv,\ell+1,v)$ or to the edges of $G$.
    Note that there are only $\ell+2$ tokens on $E(uv,\ell+1,u) \cup E(uv,\ell+1,v)$, so, at least $2(\ell+1)-(\ell+2) = \ell$ tokens have to move onto edges of $G$. 
    So, if we let $G_t$ be the graph induced by the tokens on $G$ at step $t$, we have $\abs{E(G_t)} \geq \ell$ and that $G_t$ satisfies $\Pi$ as required, since $G_t$ is a token-induced subgraph in the reconfiguration sequence and $\Pi$ is minor-closed.

    For the backwards direction, suppose \textsc{Decide($\Pi$)} is a yes-instance, so $G$ has a subgraph $G' \in \Pi$ with $E(G') \geq \ell$. 
    Starting with $E_s$, we can first move $\ell+1$ edges of $E(st,2(\ell+1),s)$ to $E(uv,\ell+1,u)$. This satisfies $\Pi$ since removing edge from $E(st,2(\ell+1),s)$ maintains $\Pi$ and $1$-summing with an edge also maintains $\Pi$, since $\Pi$ is a \emph{special} minor-closed graph class.
    Then, we move $\ell$ tokens on edges of $E(st,2(\ell+1),s)$ to the edges corresponding to $E(G')$, which maintains $\Pi$ since $G' \in \Pi$ and $\Pi$ is minor-closed, so removing tokens edges on $E(st,2(\ell+1),s)$ still maintains $\Pi$.
    Then, for the last token in $E(st,2(\ell+1),s)$, we move it to an edge on $E(uv,\ell+1,v)$. 
    Note that the token-induced subgraph on $H'$ can be constructed from $H \setminus \set{st}$, which satisfies $\Pi$, by $1$-summing with $2(\ell+1)$-edges on $t$, a $(\ell+1)$-subdivision on $uv$, and then removing $\ell$ edges from $E(uv,\ell+1,v)$. 
    So, this maintains $\Pi$ since these are all operations for which $\Pi$ is closed under.
    Finally, we sequentially move the $\ell$ tokens in $E(G')$ to the unoccupied edges in $E(uv,\ell+1,v)$. 
    This maintains $\Pi$ by the same reasoning above, since the token-induced subgraph on $H'$ can be constructed from $H \setminus \set{st}$ by $1$-summing with $2(\ell+1)$-edges on $t$, a $(\ell+1)$-subdivision on $uv$, and then removing appropriate number of edges from $E(uv,\ell+1,v)$. The subgraph induced by tokens on $E(G')$ maintains $\Pi$ because $\Pi$ is minor-closed.
    So, each step in the reconfiguration sequence induces a graph that satisfies $\Pi$.
    Consequently, \textsc{$\Pi$ reconfiguration} problem is \cNP-hard.
\end{proof}

\subsection{Graphs with Bounded Treewidth}\label{subsec:reconf_treewidth}

Before we prove the hardness of reconfiguring graphs with bounded treewidth, we need to show that bounded treewidth as a graph class is \emph{special}. 

\begin{restatable}{lemma}{twkstruct}\label{lem:tw_k_structure}
    $(\ast)$ Let $k \geq 2$.
    Let $G$ and $H$ be two graphs of treewidth $\leq k$ and let $u \in V(G)$, $vw \in E(G)$, and $a \in V(H)$. Then the following holds:
    \begin{enumerate}
        \item $tw(G \oplus_{u,a} H) \leq k$.
        \item For any $m \geq 1$, if $G'$ is the result of applying an $m$-subdivision of $vw$, then $tw(G') \leq k$.
    \end{enumerate}
\end{restatable}

\begin{theorem}\label{thm:treewidth_reconf_np_hard}
    \textsc{$k$-bounded treewidth reconfiguration} is \cNP-hard for any fixed $k \geq 2$.
\end{theorem}

\begin{proof}
    It is known that for any $k \geq 2$, the \textsc{spanning $k$-tree problem} is \cNP-hard \cite{cai_93}, where given a connected graph $G$, it asks whether $G$ contains a $k$-tree that uses all vertices of $G$.
    Note that by construction a $k$-tree with $n \geq k+1$ vertices has $kn-\frac{k^2+k}{2}$ edges, since we start with a $k+1$-clique, which has $\frac{k(k+1)}{2}$ edges, and for the remaining $n-k-1$ vertices, we add $k$ new edges to the graph, which means we have $\frac{k(k+1)}{2} + k(n-k-1) = kn-\frac{k^2+k}{2}$ edges in total.
    Recall any graph with treewidth at most $k$ is a subgraph of a $k$-tree \cite{nesetril_08}, so a graph with $n$ vertices, $k$-bounded treewidth, and $kn-\frac{k^2+k}{2}$ edges is necessarily a $k$-tree.
    In particular, this implies that \textsc{Decide($\Pi$)} is \cNP-hard by letting $\ell$ be $kn-\frac{k^2+k}{2}$, where $\Pi$ asks whether the graph has $k$-bounded treewidth.
    Furthermore, $\Pi$ contains a graph with at least one edge, namely a $k$-clique, and by Lemma \ref{lem:tw_k_structure}, $\Pi$ is a \emph{special} minor-closed graph class, since treewidth is a minor-closed property \cite{bodlaender_98}.
    So, by Theorem \ref{thm:minor_closed_hardness}, \textsc{$k$-bounded treewidth reconfiguration} is \cNP-hard.
\end{proof}

\subsection{Planar Graphs}\label{subsec:reconf_planar}

Here we show that \textsc{planar graph reconfiguration} is \cNP-hard by Theorem \ref{thm:minor_closed_hardness} as well.



\begin{theorem}\label{thm:planar_reconf_np_hard}
    \textsc{planar graph reconfiguration} is \cNP-hard.
\end{theorem}

\begin{proof}
    Let $\Pi$ be the class of planar graphs. It is well known that planar graphs are closed under taking minors, and one can easily verify that planar graph is \emph{special}, that is, closed under $k$-subdivision, closed under $1$-sum, and containing a graph with at least one edge. The result then follows directly from Theorem \ref{thm:minor_closed_hardness} using the fact that \textsc{Decide($\Pi$)} is \cNP-Hard \cite[p.~197]{garey_79}.
\end{proof}

\section{Reconfiguring with Extra Resources}\label{sec:technical3}

%

Given the intractability results above, we want to extend the existing model of \textsc{Subgraph Reconfiguration} to settings where these problems may become tractable. We draw inspiration from robot motion planning \cite{gao_23}, which studies the rearrangement of objects on a tabletop using a robotic arm. In that setting, objects may be temporarily moved to a ``buffer'' space outside the workspace due to physical constraints, motivating the minimum running buffer problem, which asks for the minimum buffer space that must be in use at any one time in order to rearrange the tabletop.
In a more related setting of \textsc{Bin Packing Reconfiguration} \cite{kam_25}, the authors also proposes a resource-focused notion of reconfiguration as future direction, in which one is given extra resources (e.g. additional bins) to facilitate the reconfiguration. 
Building on these ideas, we introduce the setting of subgraph reconfiguration with extra resources, where we allow extra buffer space during reconfiguration by permitting tokens to be placed into and retrieved from a buffer.
We formally define below what it means to reconfigure a subgraph with extra buffer space.

\begin{definition}\label{def:buffer_reconf}
Let $(G, E_s, E_t)$ be an instance of a \textsc{$\Pi$ reconfiguration} problem. 
We say $(G, E_s, E_t)$ is \emph{reconfigurable with $k$ extra buffer space} if there exists a sequence $E_s = E_0, E_1, \cdots, E_\ell = E_t$ where, $\abs{E_s} = \abs{E_t}$, and for every $i \in \{1, \cdots \ell\}$, $E_i$ satisfies $\Pi$ and either:
\begin{enumerate}
    \item $\abs{E_{i} \Delta E_{i-1}} = 2$ and $\abs{E_{i}} = \abs{E_{i-1}}$, or 
    \item $\abs{E_{i} \Delta E_{i-1}} = 1$ and $\abs{E_s} - k \leq \abs{E_i} \leq \abs{E_s}$
\end{enumerate}
The \emph{minimum required buffer} of an instance is the smallest $k$ for which the instance is reconfigurable with $k$ extra buffer space. 
If no such $k$ exists, we say the \emph{minimum required buffer} is infinity. 
\end{definition}
Intuitively, (1) corresponds to the usual reconfiguration step, (2) allows us to put a token into (or retrieve from) a size $k$ buffer. 

Although our extra buffer setting resembles the existing Token Addition-Removal (TAR) model \cite{ito11}, there are subtle differences. Our model is more restrictive. In TAR, tokens are added or removed one at a time, with additions permitted before removals, usually subject to lower and upper bounds on the sizes of the intermediate configurations. 
It is more restrictive in that we can only retrieve (add) as many tokens as we have already stored (removed).

With this framework, we can view some of the existing results in the work of Hanaka et al. \cite{hanaka_20} in terms of extra buffer space needed. More specifically, we have the following results, which follow almost directly from the proofs of Theorems 2, 3, and 6 of \cite{hanaka_20}.
\begin{restatable}{theorem}{buffersimple}\label{thm:buffer_simple}
    $(\ast)$ The following statements are true when the input graph is connected:
    \begin{enumerate}
        \item The minimum required buffer of any instance of \textsc{tree reconfiguration} is $0$.
        \item The minimum required buffer of any instance of \textsc{cycle reconfiguration} is $\infty$.
        \item For any sufficiently large $n$, there exists a size-$n$ instance of \textsc{path reconfiguration} with minimum required buffer $\Omega(n)$.
    \end{enumerate}
\end{restatable}

By modifying the construction used in the hardness proof of Theorem \ref{thm:pathwidth_k}, \ref{thm:treewidth_reconf_np_hard}, \ref{thm:planar_reconf_np_hard}, we also get the following result. Recall that the core idea of these proofs stems from constructing an input graph for which reconfiguring between the source and target subgraphs would require moving a large number of tokens into some special regions of the input graph. With slight modification, the construction can be used to show the minimum required buffer is $\Omega(n)$, by simply removing those special regions from the input graph and adding an appropriate number of new vertices, vertices with no incident edge, to the input graph.

\begin{restatable}{theorem}{bufferpwtw}\label{thm:buffer_pw_tw}
    $(\ast)$ For any sufficiently large $n$, the following statements are true:
    \begin{enumerate}
        \item There exists a size-$n$ instance of \textsc{connected $k$-bounded pathwidth reconfiguration} with minimum required buffer $\Omega(n)$ when the input graph is connected.
        \item There exists a size-$n$ instance of \textsc{$k$-bounded treewidth reconfiguration} with minimum required buffer $\Omega(n)$.
        \item There exists a size-$n$ instance of the \textsc{planar graph reconfiguration} with minimum required buffer $\Omega(n)$.
    \end{enumerate}
\end{restatable}

Motivated by the contrast between reconfiguring trees, which require no extra buffer space, and cycles, which cannot be reconfigured even with additional buffer space, we turn our attention to cactus graphs, or cacti, which combine the structure of both trees and cycles. 
A cacti is a graph, not necessarily connected, in which each edge belongs to at most one cycle. 
We note that Theorem~\ref{thm:minor_closed_hardness} does not apply to cacti, since they are not \emph{special}, and their reconfiguration complexity remains an interesting direction for future work.

\begin{figure}[t]
\centering
\begin{tikzpicture}[
    every node/.style={circle, draw, fill=white, inner sep=1pt, minimum size=4pt, font=\footnotesize},
    bg/.style={gray!55, line width=0.6pt},
    src/.style={blue, line width=1.6pt},
    tgt/.style={red!80!black, line width=1.6pt},
    scale=1.0
]

\begin{scope}[xshift=0cm]
    \node (r1) at (-1,1)  {$r_1$};
    \node (r2) at (1,1)   {$r_2$};
    \node (r3) at (1,-1)  {$r_3$};
    \node (r4) at (-1,-1) {$r_4$};
    \node (h)  at (0,0)   {$h$};

    \draw[bg] (r1) -- (r2) -- (r3) -- (r4) -- (r1);
    \draw[bg] (h) -- (r1); \draw[bg] (h) -- (r2);
    \draw[bg] (h) -- (r3); \draw[bg] (h) -- (r4);

    \draw[src] (r1) -- (r2); \draw[src] (r2) -- (h); \draw[src] (h) -- (r1);
    \draw[src] (r3) -- (r4); \draw[src] (r4) -- (h); \draw[src] (h) -- (r3);

    \node[draw=none, fill=none] at (0,-1.6) {$E_s$};
\end{scope}

\begin{scope}[xshift=3.6cm]
    \node (s1) at (-1,1)  {$r_1$};
    \node (s2) at (1,1)   {$r_2$};
    \node (s3) at (1,-1)  {$r_3$};
    \node (s4) at (-1,-1) {$r_4$};
    \node (g)  at (0,0)   {$h$};

    \draw[bg] (s1) -- (s2) -- (s3) -- (s4) -- (s1);
    \draw[bg] (g) -- (s1); \draw[bg] (g) -- (s2);
    \draw[bg] (g) -- (s3); \draw[bg] (g) -- (s4);

    \draw[tgt] (s2) -- (s3); \draw[tgt] (s3) -- (g); \draw[tgt] (g) -- (s2);
    \draw[tgt] (s1) -- (s4); \draw[tgt] (s4) -- (g); \draw[tgt] (g) -- (s1);

    \node[draw=none, fill=none] at (0,-1.6) {$E_t$};
\end{scope}

\end{tikzpicture}
\caption{Illustration of the instance $(G, E_s, E_t)$ on $G = W_5$ with a minimum
required buffer of $1$. The source $E_s$ is shown in {\color{blue}blue} and the
target $E_t$ in {\color{red!80!black} red}.}
\label{fig:w5_example}
\end{figure}
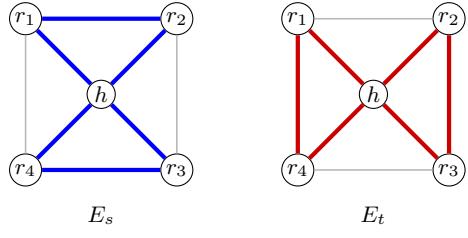

Before presenting our result, we make an observation that makes cacti an interesting candidate to study. 
Consider the instance $(G, E_s, E_t)$ where $G = W_5$ is the wheel on $5$ vertices, with rim $\set{r_1, r_2, r_3, r_4}$ and hub $h$. 
Let $E_s = \set{r_1 r_2,\, r_2 h,\, h r_1,\, r_3 r_4,\, r_4 h,\, h r_3}$ consist of two opposite triangles, and let $E_t = \set{r_2 r_3,\, r_3 h,\, h r_2,\, r_1 r_4,\, r_4 h,\, h r_1}$ consist of the other two. See Figure \ref{fig:w5_example} for an illustration.
This instance is not reconfigurable in the usual setting, but becomes reconfigurable with one extra buffer space. 
We will later show that a constant-sized buffer suffices for reconfiguring cacti whose source and target are triangular cacti, one of the simplest forms of cacti in which each cycle is a triangle and each edge is part of a triangle.

\begin{restatable}{theorem}{buffercactus}\label{thm:buffer_cactus}
    $(\ast)$ The minimum required buffer for any instance $(G, E_s, E_t)$ of \textsc{cacti reconfiguration} is at most $O(c_{max})$, where $c_{max}$ is the larger of the number of cycles in $E_s$ or $E_t$.
    Furthermore, if we restrict $E_s$ and $E_t$ to be triangular cacti, the minimum required buffer is at most $6$.
\end{restatable}

The main idea to proving the first claim is to reduce to the reconfiguration of forests, which is always possible without any extra buffer when the source and target are forests of the same graph and of equal size. 
This is because we can view the forests as independent sets of the graphic matroid, so reconfiguration follows from the symmetric exchange property for independent sets of equal size.
Removing an appropriate set of $c_{\max}$ edges from the source yields a forest subgraph $F_s$ with $\abs{E_s} - c_{\max}$ edges, and similarly removing an appropriate set of $c_{\max}$ edges from the target yields a forest subgraph $F_t$ with $\abs{E_t} - c_{\max}$ edges. 
We move the $c_{\max}$ removed edges of the source into the buffer to obtain $F_s$, reconfigure $F_s$ into $F_t$, and then
reinsert the edges in buffer to the missing target edges $E_t \setminus F_t$. 
This requires an extra buffer of size $c_{\max}$.

For the second claim, the main idea to proving it is to view a triangular cactus as a solution to a matroid parity instance (see \cite{lovasz_86, szigeti_03} for details) and to apply the \textsc{Matroid Parity Reconfiguration} result of \cite{bousquet_23}, which guarantees that two parity solutions of the same size are reconfigurable whenever a parity solution of larger size exists. 
Given an instance $(G, E_s, E_t)$, we construct $E_s'$ and $E_t'$ by removing one triangle from $E_s$ and $E_t$ respectively. 
Viewing $E_s'$ and $E_t'$ as parity solutions, the instance $(G, E_s', E_t')$ is always reconfigurable, since $E_s$ itself is a strictly larger parity solution. 
Each matroid parity reconfiguration step can be viewed as the movement of a single triangle, which can be realize using $3$ extra buffer space in our setting, storing the three tokens of the triangle in the buffer and then placing them onto the target triangle. 
Together with the $3$ buffer space holding the initially removed triangle, an extra buffer of size $6$ suffices.

\section{Conclusion}\label{sec:conclusion}

Our work showed that reconfiguration of bounded pathwidth, bounded treewidth, and planar subgraphs remains intractable. 
Along the way, we proved a general theorem for establishing hardness in \emph{special} minor-closed graph classes, which we expect to find utility beyond the results shown here. 
In light of these hardness results, we introduced the extra-buffer setting to understand how additional space can turn a non-reconfigurable instance into a reconfigurable one. 
We showed that some instances require $\Omega(n)$ buffer to reconfigure subgraphs of bounded pathwidth, bounded treewidth, and planar graphs, whereas $O(1)$ buffer suffices for reconfiguring cacti when the source and target are both triangular cacti.

Future work could identify other \emph{special} minor-closed graph classes to which Theorem \ref{thm:minor_closed_hardness} applies, and more generally, devise a similar general theorem for showing PSPACE-completeness. 
Another natural direction is to consider what other graph properties are intractable normally, but becomes reconfigurable with a constant, but non-zero, extra buffer space. 
It would be interesting to determine whether \textsc{Cacti Reconfiguration} requires only constant buffer space in general.
More broadly, we hope our extra resource setting will lead to the design of efficient algorithms parameterized by the amount of extra buffer space needed.



\bibliography{bib}

\appendix

\section{Graphs with Bounded Pathwidth}\label{app:path_width_k}

Our goal in this section is to show a complete proof of the following theorem:

\pathwidthbound*

We will organize the proof into three sections. 
First, we will provide the proof for Lemma \ref{lem:spike_hamiltonian} and introduce two technical lemmas about pathwidth in Section \ref{app:technical_lemmas} that will be used extensively later in the proof. 
In Section \ref{app:proof_setup}, we set up the necessary definitions for the proof, such as the reduction graph for the reconfiguration problem, and prove some results about them. and finally, we will proceed to prove the main theorem in Section \ref{app:main_result}.

\subsection{Technical Lemmas}\label{app:technical_lemmas}

\spikedhamiltonian*

\begin{proof}
    Note that, as a caterpillar is a tree, and $H$ uses at least one spike of each vertex in $G$, we must have that $H$ contains a spanning tree of $G$. 
    In $H$, consider following the path starting at $p_1$ until the path diverges (some vertex has degree $\geq 3$), or until the path ends (some vertex has degree $1$).
    
    Suppose for a contradiction that the path diverges at a vertex $t$ of degree $\geq 3$.
    Note that the earliest vertex of divergence is $v$.
    As the path has length $\geq 2$, let $w$ be the vertex encountered in the path prior to $t$ and note that $w$ has degree $\geq 2$.
    Let $t u_1, t u_2$ be two edges from the divergence of the path.
    Then, as each vertex of $G$ uses at least one spike, in particular $u_1$ and $u_2$ each use a spike and therefore have degree $\geq 2$.
    Then this means that $w, t, u_1, u_2$ are all in the spine of $H$.
    But the spine is a path, and $t$ still has degree $\geq 3$, a contradiction.
    
    Therefore, we must have that the path ends.
    So the part of $H$ forming a spanning tree of $G$ actually forms a path using all the vertices of $G$, with the vertex $v$ as one of its endpoints.
    So $G$ contains a Hamiltonian $v$-path. 
\end{proof}

\begin{lemma}\label{lem:pw_lower_bound}
    Let $G$ be a connected graph, $k \in \mathbb{Z}_{\geq 1}$.
    If there exists a vertex $v$ such that $G \setminus \{v\}$ results in at least $3$ connected components $C_1, C_2, C_3$ where $pw(C_i) \geq k$ for all $i \in [3]\}$, then $pw(G) \geq k+1$.
\end{lemma}

\begin{proof}
    This roughly follows the proof of Theorem 4 in \cite{scheffler_90}.
    Suppose, for a contradiction, that there exist three components $C_1$, $C_2$, $C_3$, such that $pw(C_i) \geq k$, but $pw(G) \leq k$. 
    It suffices to consider the case where $pw(G) = k$ since $C_1 \subset G$ and $pw(C_1) \geq k$.
    If the pathwidth of $pw(G) = k$, and $pw(C_1) \geq k$, $pw(C_2) \geq k$, and $pw(C_3) \geq k$, then any optimal path-decomposition for $G$ will contain a bag $B_{i_1}$ with size $k + 1$ where $B_{i_1} \subseteq V(C_1)$, and similarly for $B_{i_2} \subseteq V(C_2)$, and $B_{i_3} \subseteq V(C_3)$.
    Without loss of generality, let's assume $i_1< i_2 < i_3$.
    
    Note that $B_{i_1}$ and $B_{i_3}$ belong to the same connected component of $G \setminus B_{i_2}$, since $v \not \in B_{i_2}$.
    Then, by property $1$ of path-decomposition, there must be some vertex $u \in V(B_{i_1})$ and $w \in V(B_{i_3})$ where $\set{u,v} \subseteq B_{j_1}$ and $\set{w,v} \subseteq B_{j_2}$ for some $j_1, j_2 \in \mathbb{Z}_{\geq 1}$. 
    Specifically, we have $j_1 < i_2$ due to property $2$ of path-decomposition, because otherwise, we would have $u \in B_{i_1}$, $u \in B_{j_1}$, $u \not \in B_{i_2}$, contradicting $i_1 < i_2 < j_1$.
    Similarly, for $j_2$, we have $i_2 < j_2$, because otherwise, we would have $w \in B_{j_2}$, $w \in B_{i_3}$, $u \not \in B_{i_2}$, contradicting $j_2 < i_2 < i_3$.
    However, this is again a contradiction, because we now have $v \in B_{j_1}$, $v \in B_{j_2}$, $v \in B_{i_2}$, but $j_1 < i_2 < j_2$, which fails to satisfy property $2$ of path-decomposition.
    So, $pw(G) \geq k+1$ as desired. 
    
\end{proof}

\begin{lemma}\label{lem:pw_structural}
    Let $k \in \mathbb{Z}_{\geq 1}$.
    Let $T$ be a perfect ternary tree of depth $k$, with root $r$ and leaf vertices $\set{p_1, \cdots, p_{3^k}}$, and let $G_i$ be some connected graph for each $i \in [3^k]$.
    Define 
    \[
        G = T \oplus_{p_1, v_1} G_1 \oplus_{p_2, v_2} \cdots \oplus_{p_{3^{k}}, v_{3^{k}}} G_{3^{k}},
    \]
    where $v_i$ is some vertex in $V(G_i)$, and let 
    \[
        d_{min} := \min_{i \in [3^k]} pw(G_i) \text{ and } d_{max} := \max_{i \in [3^k]} pw(G_i),
    \]
    then the following holds:
    \begin{enumerate}
        \item $pw(G) \geq k + d_{min}$.
        \item $pw(G) \leq k + d_{max} + 1$.
    \end{enumerate}
    Furthermore, suppose there exists a $v_i$-path decomposition for each $G_i$, then the following holds as well:
    \begin{enumerate}
        \setcounter{enumi}{2}
        \item $pw(G) \leq k + d_{max}$, and there exists a $r$-path decomposition of width at most $k + d_{max}$.
        \item If $|d_{max} - d_{min}| \leq 1$, then $pw(G) = k + d_{min}$.
    \end{enumerate}
\end{lemma}

\begin{proof}
    We will prove each item individually.

    \vspace{2mm}
    \noindent
    \textbf{Point 1:} $pw(G) \geq k + d_{min}$
    \vspace{2mm}

    We proceed by induction on $k$. 
    Let $k=1$ and $r$ be the root of $T$, the depth-$1$ perfect ternary tree.
    Consider $G \setminus \{r\}$ and the three resulting components $C_1$, $C_2$, $C_3$.
    Note that the $C_i$'s are exactly the connected graphs $G_i$ that we $1$-summed to the leaves of $T$.
    Furthermore, $pw(C_i) \geq d_{min}$ by definition.
    So, by Lemma \ref{lem:pw_lower_bound}, $pw(G) \geq d_{min} + 1 = k + d_{min}$.
    Now, assume that this holds for all $k=n$ for $n \geq 1$.
    We want to show that this holds for $k = n+1$, namely, $pw(G) \leq k+1 + d_{min}$.
    Again, consider the $3$ components $C_1$, $C_2$, $C_3$ resulting from $G \setminus \{r\}$, where $r$ is the root of the depth-$k+1$ perfect ternary tree.
    By induction, we know that $pw(C_i) \geq k + d_{min}$ for all $i \in [3]$, since each $C_i$ is a perfect ternary tree of depth $k$ where its leaves are $1$-summed with graphs with pathwidth at least $d_{min}$.
    Using Lemma \ref{lem:pw_lower_bound} again, we have $pw(G) = k + d_{min} + 1 = k + 1 + d_{min}$ as required.

    \vspace{2mm}
    \noindent
    \textbf{Point 2:} $pw(G) \leq k + d_{max} + 1$
    \vspace{2mm}
    Here, we provide an explicit construction for the path-decomposition of $G$ with width $k + d_{max} + 1$.
    First, we show that any depth $k$ perfect ternary tree has a path-decomposition of width $k$.
    Let the vertices of $T$ be labelled as follows. 
    \[
        V(T) = \{v_{0,1}, v_{1,1}, v_{1,2}, v_{1,3}, v_{2,1}, \cdots, v_{2,9}, \cdots, v_{k,1}, \cdots, v_{k,3^{k}}\}
    \]
    where $v_{0,1}$ is the root node of $T$, and each $v_{i,j}$ denotes the $j$-th vertex at the $i$-th layer.
    Furthermore, the edges are defined in a natural way with respect to the vertex ordering in each layer, namely, 
    \[
        E(T) = \bigcup_{i \in [k-1]} \bigcup_{j \in [3^i]} \{ v_{i,j} v_{i+1, 3^j-2}, ~ v_{i,j} v_{i+1, 3^j-1}, ~ v_{i,j} v_{i+1, 3^j} \}.
    \]
    In particular, one such path-decomposition to consider vertices in the path in each iteration of a depth-first-search over the tree $T$.
    Concretely, we have a width-$k$ path-decomposition $(B_i: i \in [3^k])$ of $T$, where
    \begin{align*}
        B_1 &= \{v_{0,1}, v_{1,1}, v_{2,1} \cdots v_{k,1}\} \\
        B_2 &= \{v_{0,1}, v_{1,1}, v_{2,1} \cdots v_{k,2}\} \\
        B_3 &= \{v_{0,1}, v_{1,1}, v_{2,1} \cdots v_{k,3}\} \\
        B_4 &= \{v_{0,1}, v_{1,1}, v_{2,1} \cdots v_{k-2, 1}, v_{k-1, 2}, v_{k,4}\} \\
        &\vdots \\
        B_{3^k} &= \{v_{0,1}, v_{1,3}, v_{2,9} \cdots v_{k-2, 3^{k-2}}, v_{k-1, 3^{k-1}}, v_{k,3^{k}}\}
    \end{align*}
    
    For each $i \in [3^k]$, since $pw(G_i) \leq d_{max}$ by definition, we know that there exists a path-decomposition of $G_i$ of width $d_{max}$, say $(D_{i,1}, \cdots, D_{i, m_i})$ for some $m_i \in \mathbb{Z}$, where each $D_{i,j} \subseteq V(G_i)$ with $|D_{i,j}| \leq d_{max}+1$.
    Then, to get the desired path-decomposition of $G$, we can glue the decompositions of $T$ and each $G_i$ together to get $(B_{i,j}: i \in [3^k], j \in [m_i])$, where 
    \begin{equation}\label{eqt:bag_ij}
        B_{i,j} = B_i \cup D_{i,j}.
    \end{equation} 
    Here, we implicitly assume an ordering of indices where $(i,j) < (i', j')$ when $i < i'$ or when $i = i'$ and $j < j'$.
    
    We will show that this is a valid path-decomposition. 
    Note that for each edge $uv \in E(G)$, we have $uv \in E(T)$ or $uv \in E(G_i)$ for some $i \in [3^k]$. 
    Either way, $\{u,v\}$ will be included in some $B_i$ or $D_{i,j}$, which means it is included in some $B_{i,j}$ by construction, so the first property of Definition \ref{def:pathwidth} is satisfied. 
    For the second property, we suppose, for the contrary, that there exists some indices $(i_1, j_1), (i_2, j_2), (i_3, j_3)$, where $(i_1, j_1) < (i_2, j_2) < (i_3, j_3)$, and some $v \in V(G)$ such that $v \in B_{i_1, j_1}$ and $v \in B_{i_3, j_3}$, but $v \not\in B_{i_2, j_2}$. 
    There are two cases to consider. 
    If $v \in V(T)$, then this would imply $v \in B_{i_1}$, $v \not\in B_{i_2}$, and $v \in B_{i_3}$, which contradicts that $(B_i: i\in[3^k])$ is a valid path-decomposition of $T$.
    Otherwise, if $v \in V(G_i) \setminus V(T)$ for some $i \in [3^k]$, then this would imply $v \in D_{i_1, j_1}$, $v \not\in D_{i_2, j_2}$, and $v \in D_{i_3, j_3}$. Note that $D_{i,j} \cap D_{i',j'} = \emptyset$ when $i \not= i'$ as they are bags corresponding to path-decompositions of different graphs, which implies $i_1 = i_2 = i_3$.
    However, this now contradicts the fact $(D_{i_1,j}: j \in [m_{i_1}])$ is a path-decomposition of $G_{i_1}$.
    So, the second property is satisfied as well. The third property is clear by construction, and so, $(B_{i,j}: i \in [3^k], j \in [m_i])$ is a valid path-decomposition. Moreover, since 
    \begin{equation}\label{eqt:path_decomp_inequality_1}
        \max_{i,j} |B_{i,j}| \leq \max_{i,j} |B_i \cup D_{i,j}| \leq \max_{i,j} |B_i| + |D_{i,j}| \leq k + d_{max} + 2,        
    \end{equation}
    our decomposition has width at most $k + d_{max} + 1$.

    \vspace{2mm}
    \noindent In the following points, we will add the assumption that \textit{there exists a $v_i$-path decomposition for each $G_i$.}
    
    \vspace{2mm}
    \noindent
    \textbf{Point 3:} $pw(G) \leq k + d_{max}$
    \vspace{2mm}
    
    We will show this by a very similar construction to the one in point $2$. Specifically, instead of constructing each bag $B_{i,j} = B_{i} \cup D_{ij}$ as in Equation \ref{eqt:bag_ij}, we utilize the assumption that there exists a $v_i$-path decomposition for each $G_i$.
    In particular, we define each $B_{i,j}$ as follows
    \begin{equation}\label{eqt:bag_ij_optimized}
        B_{i,j} = (B_{i} \setminus \{v_i\}) \cup D_{ij}.
    \end{equation}
    
    We will show that this is a valid path-decomposition. 
    Note that for each edge $uv \in E(G)$, we have $uv \in E(T)$ or $uv \in E(G_i)$ for some $i \in [3^k]$. 
    Suppose $v_i \in \{u,v\}$ for some $i \in [3^k]$, and without loss of generality, $u = v_i$.
    If $v \in V(G_i)$, then $uv$ is included in $D_{i,1}$, otherwise, if $v \in V(T)$, then $uv \in B_i$ by construction.
    Note that $B_{i,1} = (B_{i} \setminus {v_i}) \cup D_{i,1}$, and by assumption $v_i \in D_{i,1}$, we have $B_i \subseteq B_{i,1}$ and $D_{i,1} \subseteq B_{i,1}$, which implies $uv \in B_{i,1}$.
    Now, suppose $v_i \not \in \{u,v\}$ for any $i \in [3^k]$, the analysis in point (2) follows, where $\{u,v\}$ will be included in some $B_i \setminus \{v_i\}$ or $D_{i,j}$, so $\{u,v\}$ is included in some $B_{i,j}$ by construction.
    Hence, the first property of Definition \ref{def:pathwidth} is satisfied.
    
    For the second property, we suppose, for the contrary, that there exists some indices $(i_1, j_1), (i_2, j_2), (i_3, j_3)$, where $(i_1, j_1) < (i_2, j_2) < (i_3, j_3)$, and some $v \in V(G)$ such that $v \in B_{i_1, j_1}$ and $v \in B_{i_3, j_3}$, but $v \not\in B_{i_2, j_2}$. 
    There are three cases to consider here. The first two is same as the previous analysis in point (2).
    \begin{enumerate}[\alph*]
        \item If $v \in V(T) \setminus \{v_i: i \in [3^k]\}$, then this would imply $v \in B_{i_1}$, $v \not\in B_{i_2}$, and $v \in B_{i_3}$, which contradicts that $(B_i: i\in[3^k])$ is a valid path-decomposition of $T$.
        \item If $v \in V(G_i) \setminus \{v_i: i \in [3^k]\}$ for some $i \in [3^k]$, then this would imply $v \in D_{i_1, j_1}$, $v \not\in D_{i_2, j_2}$, and $v \in D_{i_3, j_3}$. Note that $D_{i,j} \cap D_{i',j'} = \emptyset$ when $i \not= i'$ as they are bags corresponding to path-decompositions of different graphs, which implies $i_1 = i_2 = i_3$.
        However, this now contradicts the fact $(D_{i_1,j}: j \in [m_{i_1}])$ is a path-decomposition of $G_{i_1}$.
        \item If $v \in \{v_i: i \in [3^k]\}$, then this would imply $v \in B_{i_1, j_1}$, $v \not\in B_{i_2, j_2}$, and $v \in B_{i_3, j_3}$. Note that by construction, if $v \in B_{i,j}$, then $v \not \in B_{i',j'}$ for any $i' \not= i$ and $j' \in [m_{i'}]$, and so, $i_1 = i_2 = i_3$. 
        Let $i = i_1$.
        By construction, if $v$ is in $B_{i,j}$ for any $j \in [m_{i}]$, $v$ must also be in $D_{i,j}$.
        However, $v \in D_{i, j_1}$, $v \not\in D_{i, j_2}$, and $v \in D_{i, j_3}$ would then contradict that $(D_{i,j}: j \in [m_{i}])$ is a path-decomposition of $G_{i}$.
    \end{enumerate}
    So, the second property is satisfied as well, and thus, $(B_{i,j}: i \in [3^k], j \in [m_i])$ is a valid path-decomposition. 
    Moreover, since 
    \begin{equation}\label{eqt:path_decomp_inequality_2}
        \max_{i,j} |B_{i,j}| \leq \max_{i,j} |B_i \setminus \{v_i\} \cup D_{i,j}| \leq \max_{i,j} |B_i-1| + |D_{i,j}| \leq k + d_{max} + 1,        
    \end{equation}
    our decomposition has width at most $k + d_{max}$.
    Furthermore, this is a $r$-path decomposition, where $r$ is the root of $T$, since $r \in B_1 \subseteq B_{1,1}$ and $r \not= v_1$.

    \vspace{2mm}
    \noindent
    \textbf{Point 4:} If $|d_{max} - d_{min}| \leq 1$, then $pw(G) = k + d_{min}$
    \vspace{2mm}
    
    Let's start with the simple case where $d_{max} = d_{min}$. Since all $G_i$ has a $v_i$-path decomposition, by point $(3)$, we have $pw(G) \leq k + d_{max} = k + d_{min}$. Combine this with point $(2)$, we have $pw(G) = k + d_{min}$ as desired.
    
    Now, on to the case where $d_{max} = d_{min} + 1$. We will first prove the following claim.
    
    We will prove this by induction on $k$.
    Let $k = 1$ and $r$ be the root of $T$, and $r_1, r_2, r_3$ be the three vertices adjacent to $r$.
    In particular, $r_1, r_2, r_3$ are also the leaf vertices of $T$.
    Without loss of generality, suppose $pw(G_1) = d_{min}$ and $pw(G_3) = d_{max}$, while $pw(G_2)$ can be either.
    By assumption, there is a $v_i$-path decomposition for each $G_i$.
    Let $(B_{1,1} \cdots B_{1,m_1})$, $(B_{2,1} \cdots B_{2,m_2})$, $(B_{3,1} \cdots B_{3,m_3})$ be the $v_1, v_2, v_3$-path decompositions for $G_1, G_2, G_3$ respectively.
    Note that by construction of $G$, $v_1, v_2, v_3$ can also be identified as $r_1, r_2, r_3$ due to the $1$-sum operation.
    Then, we have the following path-decomposition $(B_1, \cdots B_m)$ for $G$ where
    \begin{alignat*}{3}
        B_1 &= B_{2,m_2} \qquad & B_{m_2+1} &= \{r,r_2\} \qquad & B_{m_1+m_2+2} &= \{r,r_3\}  \\
        B_2 &= B_{2,m_2-1} \qquad & B_{m_2+2} &= \{r\}\cup B_{1,1} \qquad & B_{m_1+m_2+3} &= B_{3,1}  \\
        &\;\; \vdots && \;\;\vdots && \;\;\vdots \\
        B_{m_2} &= B_{2,1} \qquad & B_{m_1+m_2+1} &= \{r\}\cup B_{1,m_1} \qquad & B_{m_1+m_2+m_3+2} &= B_{3,m_3}.
    \end{alignat*}
    We show that this is a valid path-decomposition (Definition \ref{def:pathwidth}). 
    For the first property, consider any edge $uv \in E(G)$. 
    If $r \in \{u,v\}$, then $uv$ is either $rr_1$, $rr_2$, or $rr_3$. 
    Since $rr_2 \in B_{m_2 + 1}$, $rr_3 \in B_{m_1+m_2+2}$, and $rr_1 \in B_{j}$ for some $j \in [m_2+2, m_1+m_2+1]$, we must have $uv \in B_{i}$ for some $i \in [3^k]$.
    If $r \not \in \{u,v\}$, then $uv \in G_i$ for some $i \in [3]$, which means it is in the path-decomposition of $G_i$, so it is in the path-decomposition of $G$ by construction.
    For the second property, we can easily verify that since $r_2 \in B_{2,1} = B_{m_2}$ (adjacent to the bag $B_{m_2+1} = \{r, r_2\}$), $r_3 \in B_{3,1} = B_{m_1 + m_2 + 3}$ (adjacent to the bag $B_{m_1+m_2+2} = \{r, r_3\}$), and $r \in B_j$ for all $j \in [m_2+1, m_1+m_2+2]$ but nowhere else.
    Note that $\max_{i} |B_i| = \max\{2, d_{max}+1, d_{min}+2\} = d_{min} + 2$, where the second equality follows from $2 \leq d_{min} + 2$ and $d_{min} + 1 = d_{max}$.
    So, the path decomposition above has width $d_{min} + 1 = k + d_{min}$ as required.
    
    Now, assume this holds for all $k = n$ for $n \geq 1$.
    We want to show that this holds for $k = n+1$.
    By assumption, there must be at least one $G_i$ with $pw(G_i) = d_{min}$ with a $v_i$-decomposition for some $i \in [3^k]$.
    Consider $G \setminus \set{r}$, where $r$ is the root of the depth-$k$ ternary tree $T$.
    Let $r_1, r_2, r_3$ be the three vertices adjacent to $r$ and we let $C_1, C_2, C_3$ be the three components resulting from $G \setminus \set{r}$, with $r_1 \in V(C_1)$, $r_2 \in V(C_2)$, and $r_2 \in V(C_2)$, respectively. 
    Without loss of generality, assume $v_i$ belongs to the first component $C_1$.
    By point $(3)$, we have $pw(C_2) \leq (k-1) + d_{max} = k + d_{max} - 1 = k + d_{min}$ with a $r_2$-path decomposition of $C_2$ of width at most $k + d_{min}$, say $(B_{2,1} \cdots B_{2,m_2})$.
    Similarly, we have $pw(C_3) \leq k + d_{min}$ with a $r_3$-path decomposition of $C_3$ of width at most $k + d_{min}$, say $(B_{3,1} \cdots B_{3,m_3})$.
    By induction, we have $pw(C_1) = (k-1) + d_{min} = k + d_{min} - 1$, and thus, a path-decomposition of width $k + d_{min} - 1$, say $(B_{1,1} \cdots B_{1,m_1})$.
    
    Then, using a similar construction as the base case, with each $B_{i,j}$ replaced with the new definition, we have the following path-decomposition for $G$,
    \begin{alignat*}{3}
        B_1 &= B_{2,m_2} \qquad & B_{m_2+1} &= \{r,r_2\} \qquad & B_{m_1+m_2+2} &= \{r,r_3\}  \\
        B_2 &= B_{2,m_2-1} \qquad & B_{m_2+2} &= \{r\}\cup B_{1,1} \qquad & B_{m_1+m_2+3} &= B_{3,1}  \\
        &\;\; \vdots && \;\;\vdots && \;\;\vdots \\
        B_{m_2} &= B_{2,1} \qquad & B_{m_1+m_2+1} &= \{r\}\cup B_{1,m_1} \qquad & B_{m_1+m_2+m_3+2} &= B_{3,m_3}.
    \end{alignat*}
    The analysis of it being a valid path-decomposition is almost identical to the one above.
    Note that $\max_{i} |B_i| = \max\{2, k + d_{min} + 1\} = k + d_{min} + 1$, where the second equality follows from $k \geq 1$. So, the path decomposition above has width $k + d_{min}$ as required.
    
\end{proof}

\subsection{Proof Setup}\label{app:proof_setup}
\begin{definition}
    For $k, l \geq \mathbb{Z}_{\geq 1}$, we can define a \emph{$(k,l)$-claw} $C$ as a perfect ternary tree $T$ of depth $k$ with root $r$. where, for each leaf $v$, we $1$-sum it with a path of length $l$ at its ends. 
    Furthermore, we say each of those paths as a tail of $C$ and $r$ as the root of $C$.
    An \emph{almost $(k,l)$-claw} is a $(k,l)$-claw with one tail removed.
\end{definition}

\begin{proposition}\label{prop:claw_size}
    A $(k,l)$-claw has $\frac{3^{k+1}-1}{2} + 3^kl -1$ edges.
\end{proposition}

\begin{proof}
    There are $(\sum_{i=1}^{k}3^i)$ edges in the perfect ternary tree and $3^k$ paths with $l$ edges, so, by expanding the sum of the geometric series, we have 
    \[
        \sum_{i=1}^{k}3^i + 3^kl = \frac{3^{k+1}-1}{2} + 3^kl - 1
    \]
    edges in total.
    
\end{proof}

\begin{corollary}\label{cor:properties_of_claw}
    For $k, l \in \mathbb{Z}_{\geq 1}$, the following statements about pathwidth are true:
    \begin{enumerate}
        \item A perfect ternary tree of depth $k$ has pathwidth $k$.
        \item A $(k,l)$-claw has pathwidth $k+1$.
        \item An almost $(k,l)$-claw has pathwidth $k$.
    \end{enumerate}
\end{corollary}

\begin{proof}
    We will address each item individually.
    \begin{enumerate}
        \item We use Lemma \ref{lem:pw_structural} by considering each $G_i$ as a singleton graph (which has pathwidth $0$), then by point $1$ and $3$ in Lemma \ref{lem:pw_structural}, the pathwidth is $k$.
        \item A $(k,l)$ claw $C$ consists of a perfect ternary tree of depth $k$ where each leaf is $1$-summed with a path of length $l$ at one of its end vertex. 
        Since a length $l$ path has pathwidth $1$, and it has a path decomposition of width $1$ starting with the end vertex, so by points $1$ and $3$ of Lemma \ref{lem:pw_structural}, $pw(C) = k + 1$.
        \item An almost $(k,l)$ claw consists of a perfect ternary tree of depth $k$ where each leaf is $1$-summed with a path of length $l$, except for one leaf, say $v$. To use Lemma \ref{lem:pw_structural}, we can view an almost $(k,l)$-claw as $1$-summing a singleton graph to $v$ and a length $l$ path for all other leaves. This means $d_{min} = 0$ and $d_{max} = 1$, so by point 3 Lemma \ref{lem:pw_structural}, we have pathwidth equals to $k$. 
    \end{enumerate}
\end{proof}

Here, we define the reduction graph that forms the input graph for the \textsc{connected $k$-bounded pathwidth reconfiguration} problem in the reduction.

\begin{definition}\label{def:aux-pathwidth}
    Let $G = (V, E)$ be a non-empty graph with $v \in V(G)$ and let $n = \abs{V(G)}$.
    We define the reduction graph $G_{path}(v)$ as follows:
    \begin{itemize}
        \item Let $g(k,l) = 3^{k-1}l^3$ and $f(k,l) = 3^{k-1} \cdot (l-1 + l \cdot g(k,l)) + 1$.
        \footnote{The parameters are made somewhat loose for improved readability.}
        \item Let $C$ be a $(k, f(k,n))$-claw with root $r$ and leaf vertices $\{t_1, \cdots, t_{3^k}\}$.
        \item For all $i \in [3^k]$, define a new graph $G_i$ to be a copy of $G$, where $v_{G_i} \in V(G_i)$ corresponds to $v$ of $G$, and let $S_i$ be a $g(k,n)$-spiked graph of $G_i$ with $v_{i} \in V(S_i)$ corresponding to $v_{G_i}$ of $G_i$.
    \end{itemize}
    Then, $G_{path}(v)$ is formed by performing $C \oplus_{t_{i}, v_{i}} S_i$ iteratively for each $i \in [3^k]$.
\end{definition}

\begin{figure}[!ht]
    \centering
    \includesvg[width=0.7\textwidth]{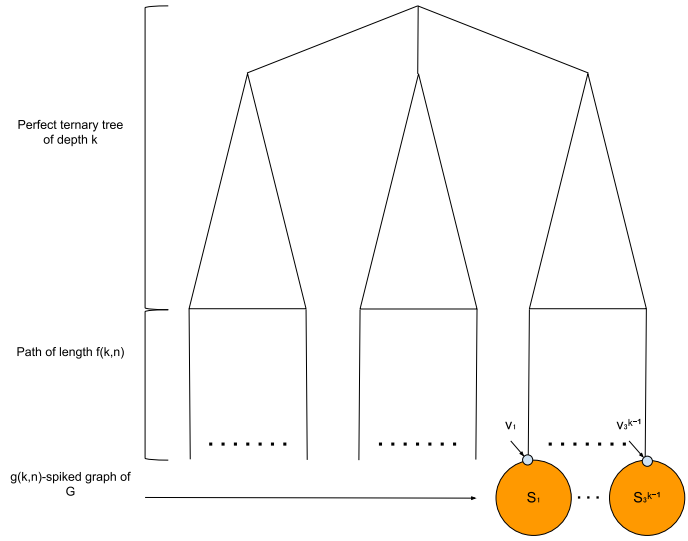}
    \caption{Example of the reduction graph $G_{path}(v)$}
\end{figure}

\subsection{Main Result}\label{app:main_result}

\pathwidthbound*

\begin{proof}
    To prove it is $\cNP$-hard, we show a polynomial-time reduction from the \textsc{Hamiltonian $v$-path} problem to the \textsc{connected $k$-bounded pathwidth reconfiguration} problem.
    As mentioned before, it suffices to consider the case where $k \geq 2$, since Theorem \ref{thm:caterpillar} has already addressed the case $k=1$.
    Let $(G, v)$ be an instance of the \textsc{Hamiltonian $v$-path} problem, where $G$ is non-empty.
    Using $G$, we create the reduction graph $G_{path}(v)$ from Definition \ref{def:aux-pathwidth}.
    Let $T$ be the ternary tree in $G_{path}(v)$ and let $P_i$ denote the tail in the $(k, f(k,n))$-claw $C$ that contains $t_i$, for all $i \in [3^k]$. 
    We write $P_{i} = p_{0}^{i} \cdots p_{f(k,n)}^{i}$ to be length $f(k,n)$ path $1$-summed to the leaf $t_i \in V(T)$, and in particular, we assume $p_{f(k,n)}^{i}$ is the vertex $1$-summed with $t_i$. 
    Let
    \[
        E_s = E(C) \setminus E(P_{3^{k-1}+1}) \text{ and } E_t = E(C) \setminus E(P_{1}),
    \]
    our goal is to show that $(G_{path}(v), E_s, E_t)$ is reconfigurable if and only if $G$ has a Hamiltonian path starting at $v$.
    When convenient, we refer to the path $P_1$ as $P_{S \setminus T}$ and $P_{3^{k-1} + 1}$ as $P_{T \setminus S}$, since those are the paths in $E_s \setminus E_t$ and $E_t \setminus E_s$ respectively.
    
    But first, we need to prove a claim that will be useful in the forward direction of the proof.
    \begin{figure}[!tbp]
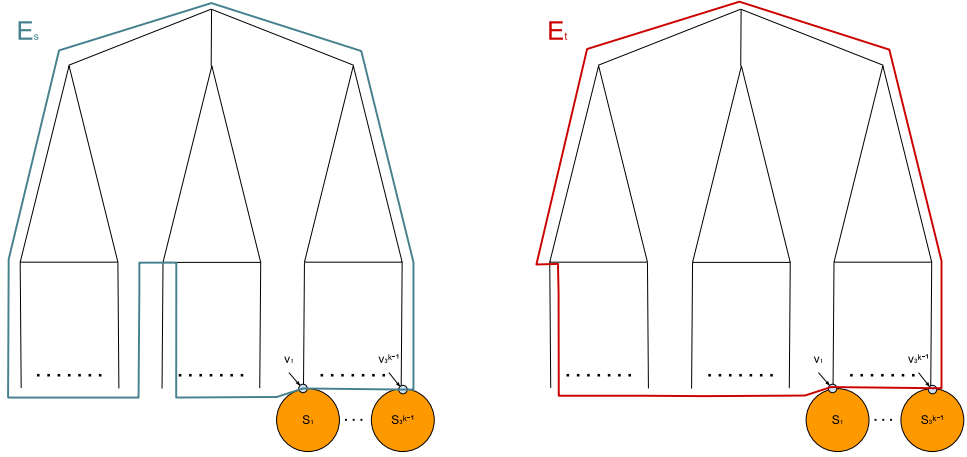

        \centering
        \subfloat{\includesvg[width=0.45\textwidth]{figures/pathwidth_reduction_source.svg}\label{fig:pathwidth_source}}
        \qquad
        \subfloat{\includesvg[width=0.45\textwidth]{figures/pathwidth_reduction_target.svg}\label{fig:pathwidth_target}}
        \caption{Source subgraph \textcolor{blue}{$E_s$} (left) and target subgraph \textcolor{red}{$E_t$} (right) in $G_{path}(v)$.}
    \end{figure}

    \begin{claim}\label{claim:pw_count}
        Let $M$ be the total number of edges in all of $S_i$ except for the spikes at one of the vertices, namely $M = \left(\sum_{i=1}^{3^{k-1}} \abs{E(S_i)}\right) - g(k,n)$. Then, $M < f(k,n)-1$.
    \end{claim}
    
    \begin{proof}
        Note that since $n \geq 1$,
        \begin{align*}
            (f(k,n)-1) - M &= 3^{k-1} \left[(n-1 + n \cdot g(k,n)) - (n \cdot g(k,n) + \abs{E(G)}\right] + g(k,n) \\
            &\geq 3^{k-1} \left[(n-1 + n \cdot g(k,n)) - (n \cdot g(k,n) + \abs{E(K_n)}) \right] + g(k,n) \\ 
            &= 3^{k-1}[n-1 - \abs{E(K_n)}] + g(k,n) \\
            &= 3^{k-1}(n-1) - 3^{k-1}\frac{n(n-1)}{2} + 3^{k-1}n^3 \\
            &= 3^{k-1}(n^3 - (\frac{n(n-1)}{2} - n + 1)) > 0,
        \end{align*}
        where $K_n$ is the complete graph with $n$ vertices.
        So, we have $f(k,n) - 1 > M$ as desired.
        
    \end{proof}

    \noindent With this, we are finally ready to prove the main result. 
    \begin{description}
        \item[$(\Leftarrow)$] 
        Suppose that $G$ has a Hamiltonian path starting at $v$. 
        Then, for each $1 \leq i \leq 3^{k-1}$, let $e_{1}^{i} \cdots e_{n-1}^i$ denote the edges on the Hamiltonian path of $G_i$ in $S_i$.
        This is possible since each $S_i$ is a spiked graph of $G_i$, which is a copy of $G$.
        Note that the number of edges on each Hamiltonian path is $n - 1$ and that the number of spikes for each $S_i$ is $n \cdot g(k,n) = n \cdot 3^{k-1}n^3$.
        Furthermore, note that each of the tail $P_i$ has $f(k,n) = 3^{k-1} \cdot (n - 1 + n \cdot g(k,n)) + 1$ edges.
        
        Define the following variables:
        \begin{itemize}
            \item $m_H := n - 1$ for the size of Hamiltonian $v$-path
            \item $m_S := n \cdot g(k,n) = n \cdot 3^{k-1}n^3$ for the total number of spikes in each $S_i$
            \item $m := m_H + m_S$.
        \end{itemize}
        Then, one possible reconfiguration sequence from $E_s$ to $E_t$ in $G_{path}(v)$ is to move all the edges in path $P_1$ to each of the $S_i$'s while maintaining the connected $k$-bounded pathwidth property.
        \begin{enumerate}
            \item Let $i = 1$. While $i \leq 3^{k-1}$, we want to move $m$ tokens to each $S_i$.
                \begin{enumerate}
                    \item Let $a = (i - 1) \cdot m$.
                    \item First, we move the token $p_{a}^{1} p_{a + 1}^{1}$ to $e_1^{i}$, $p_{a + 1}^{1} p_{a + 2}^{1}$ to $e_2^{i}$, and so on, until $p_{a + m_H -1}^1 p_{a + m_H}^1$ to $e_{m_H}^{i}$. This is essentially moving the tokens on the tail $P_1$ to the edges on the Hamiltonian $v$-path in $S_i$.
                    \item Then, we move each of $p_{a + m_H}^1 p_{a + m_H + 1}^1, \cdots, p_{a + m - 1}^1 p_{a + m}^1$ to each of the spikes in $S_i$ sequentially. After this step, there should be a token on all the spikes in $S_i$.
                    \item Let $i = i + 1$.
                \end{enumerate}
            \item We move the token from $p_{f(k,n)-1}^1 p_{f(k,n)}^1$ to $p_{f(k,n)-1}^{3^{k-1} + 1} p_{f(k,n)}^{3 ^{k-1} + 1}$. This is moving the last token on the tail $P_{S \setminus T}$ to the tail $P_{T \setminus S}$, that is, the tail that is in the target but not in the source configuration.
            \item Finally, we symmetrically move all the tokens that we moved to each $S_i$ in previous steps back to the tail $P_{T \setminus S}$.
         \end{enumerate}
   
        We claim that each step of the reconfiguration sequence induces a graph that satisfies the property \textsc{connected $k$-bounded pathwidth}. 
        It suffices to consider token jumps performed in step $a$ of our reconfiguration sequence, since the remaining steps are symmetric in the sense that the graph induced is isomorphic to a graph induced by a previous reconfiguration step. 
        It is easy to see that at each step, the graph induced by the tokens is connected.
        Note that $E_s$ and $E_t$ both induce a pathwidth $k$ graph by Corollary \ref{cor:properties_of_claw}, since they are an almost $(k,l)$-claw.

        Now, consider the $j^{th}$ token jump in step $(a)$ for some $1 \leq j \leq 3^{k-1} \cdot m$. 
        Let $G'$ be the subgraph induced by the tokens on $G$ right after the $j^{th}$ token jump.
        Note that $G'$ is a perfect ternary tree of depth $k$ where each leaf node is $1$-summed with a graph of pathwidth $1$, except for the leaf $t_{3^k-1}+1$, where we can consider it as $1$-summed with the singleton graph, which has pathwidth $0$.
        For each tail $P_i = p_{0}^{i} \cdots p_{f(k,n)}^{i}$, the graph induced by the tokens on $P_i$ forms a path with $p_{f(k,n)}^{i}$ as one of its endpoints (or a singleton, in the case of $P_{3^{k-1}+1}$), which means it has a $p_{f(k,n)}^{i}$-path decomposition.
        Furthermore, by construction, each leaf nodes $t_{i}$ for $i \in \set{2 \cdot 3^{k+1}+1 \cdots 3^k}$ is $1$-summed with a caterpillar with $p_{f(k,n}^{i}$ as one of its ends.
        This is true since extending the tail $P_i$ with a Hamiltonian path remains a path (e.g. step $(b)$), and viewing it as a caterpillar, it has $p_{f(k,n}^{i}$ as one of its ends. 
        We then add spikes to the path in step $(c)$, but this does not change the ends of the caterpillar.
        It is not hard to see that for any caterpillar $C$, with $e \in V(C)$ as one of its ends, there is an $e$-path decomposition of $C$. 
        
        Note that each leaf node $t_i$ is also identified as $p_{f(k,n)}^{i}$ due to the $1$-sum operation.
        So, from the analysis above, we have that each of the graphs $1$-summed with the leaf node $t_i$ has a $t_i$-path decomposition.
        By point $(4)$ of Lemma \ref{lem:pw_structural}, we have that $pw(G') = k + d_{min} = k$, so each subgraph induced in our reconfiguration sequence has pathwidth $\leq k$.
        Hence, our sequence is a valid reconfiguration sequence, and so $(G_{path}(v), E_s, E_t)$ is reconfigurable.

        \item[$(\Rightarrow)$] Suppose that $(G_{path}(v), E_s, E_t)$ is reconfigurable.
        Then, for any reconfiguration sequence, there exists a step during the reconfiguration in which a token is placed on the edge $e' = p_{f(k,n)}^{3^{k-1}+1} p_{f(k,n)-1}^{3^{k-1}+1}$. 
        This is the first edge of the tail $P_{T \setminus S}$.
        Note that each token-induced subgraph in our reconfiguration sequence has to be connected by definition of \text{connected $k$-bounded pathwidth}.
        
        Consider the earliest step $t$ in the sequence where a token is placed on $e'$.
        Since $t$ is the earliest step, by the connectivity constraint, no tokens have been placed on $P_{T \setminus S}$ before step $t$.
        In particular, for all reconfiguration steps before step $t$, the tokens that were moved away from the $(k,l)$-claw must be placed on edges in the spiked graphs, which contain $\sum_{i=1}^{3^{k-1}} \abs{E(S_i)})$ edges. 
        We claim the following:
        \begin{claim}\label{claim:pw_ternary_tree}
            For any step of the reconfiguration sequence before step $t$, the token-induced subgraph contains the depth $k-1$ ternary tree rooted at $r$ (the same root of the depth $k$ ternary tree).
        \end{claim}
        \begin{proof}
            Suppose not. Then, it means there is at least a token on a leaf edge of the depth $k-1$ ternary tree that has been moved elsewhere.
            By the connectivity constraint, we need to move away all tokens on at least $2$ tails, since $1$ tail could be the one remove by construction of the almost $(k,l)$-claw in $E_s$.
            However, we know that $\sum_{i=1}^{3^{k-1}} \abs{E(S_i)}) < 2f(k,n)$ by Claim \ref{claim:pw_count}, since $\sum_{i=1}^{3^{k-1}} \abs{E(S_i)}) < f(k,n) - 1 + g(k,n) < 2f(k,n)$.
            This is a contradiction since we don't have enough space to move away all those tokens. 
        \end{proof}
        Now, suppose, for the contrary, that at step $t$, there exists a token on some edge on each of the tails $P_i$, for $i \in [3^k]$. 
        By Claim \ref{claim:pw_ternary_tree}, we know each step on the sequence before step $t$ contains the depth $k-1$ ternary tree rooted at $r$.
        So, by the connectivity constraint, these tokens on the $3^k$ tails, along with the depth $k-1$ ternary tree, induce a graph which contains a $(k,1)$-claw as a subgraph, so it has pathwidth $\geq k+1$ by Corollary \ref{cor:properties_of_claw}, contradicting the pathwidth $\leq k$ constraint.
        This means, at step $t$, there must be a tail $P'$ where there is no token placed on it.
        This implies, at step $t-1$, right before moving the token to $e'$, at least $f(k,n)-1$ tokens on $P'$ have already been moved to some other edges on $G_{path}(v)$.
        
        Due to the connectedness property, we cannot move any tokens to edges of $P_{T \setminus S} \setminus \set{e'}$ before first moving a token to $e'$.
        Since step $t$ is the first step for which $e'$ has a token placed on it, the only place for the tokens on $P'$ to move to is the edges on each of the $g(k,n)$-spiked graph, namely $S_i$ for $i \in [3^{k-1}]$. 
        Since we have moved away at least $f(k,n)-1$ tokens of $P'$ away at step $t-1$, by Claim \ref{claim:pw_count}, there is at least a token on some spike of $v$ for every vertex $v \in \bigcup_{i=1}^{3^{k-1}} V(G_i)$.
        
        We make another observation. 
        Let $U$ be the subgraph induced by the tokens at step $t$. 
        For $i \in [3^{k-1}]$, let $R_i$ denote the token-induced subgraph constrained to the edges in the $g(k,n)$-spiked graph $S_i$ and the tail $P_{2 \cdot 3^{k-1} + i}$. 
        Then $U$ and all of the $R_i$'s are non-empty connected subgraphs of $G_{path}(v)$.
        
        Recall that $r$ is the root of the ternary tree $T$, with adjacent vertices $\set{r_1, r_2, r_3}$.
        Let $C_1, C_2, C_3$ be the connected components we get from $G_{path}(v) \setminus \{r\}$, where $r_i \in C_i$ for $i \in [3]$.
        
        Now, consider the edge $rr_3$. 
        Suppose, for the contrary, that there is no token on $rr_3$ at step $t$ of the reconfiguration sequence. 
        Then, since $U$ is connected, there are $2$ possibilities.
        \begin{description}
            \item[Case 1:] There are no tokens placed on any of the edges of $C_3$.
            Recall that each $R_i$, for $i \in [3^{k-1}]$, are non-empty.
            Since the $R_i$'s are subgraphs of $C_3$, this is not possible.
            \item[Case 2:] There are no tokens placed on any of the edges of $E(C_1) \cup E(C_2) \cup \set{rr_1, rr_2}$.
            Note that at any point in the reconfiguration, there is the same number of tokens on $G_{path}(v)$.
            In the very least, there should be no tokens placed on $C_1 \cup \set{rr_1}$, so all tokens that are initially on $C_1$ and $rr_1$ have to move elsewhere on $G_{path}(v)$.
            There are $\abs{E(C_1) \cup \set{rr_1}} = \frac{3^{(k-1)+1}-1}{2} + 3^{k-1} f(k,n) - 1 + 1 = \frac{3^{k}-1}{2} + 3^{k-1} f(k,n)$ edges on $C_1$ and $rr_1$ initially.
            However, there are at most 
            \begin{align*}
                \abs{E(P_{T \setminus S})} + \sum_{i=1}^{3^{k-1}}\abs{E(S_i)} &= f(k,n) + \sum_{i=1}^{3^{k-1}} (\abs{E(G_i)} + n \cdot g(k,n)) \\
                &\leq f(k,n) + 3^{k-1} \cdot (n^2 + n \cdot g(k,n)) \\
                &= f(k,n) +  3^{k-1} \cdot [(n-1 + n \cdot g(k,n)) + (n^2-n+1)] \\
                &= f(k,n) + f(k,n) + 3^{k-1} (n^2-n+1) \\
                &\leq 2f(k,n) + (3^{k-1} n^3) \\
                &= 2f(k,n) + g(k,n) \\
                &\leq 3f(k,n)
            \end{align*}
            slots for the tokens to move to.
            Furthermore, note that 
            \begin{align*}
                \abs{E(B'_1) \cup \set{rv_1}} &= \frac{3^{k}-1}{2} + 3^{k-1} f(k,n) \\
                &\geq 4 + 3^{k-1}f(k,n) \tag{since $k \geq 2$} \\
                &> 3f(k,n).
            \end{align*}
            So, this case is also impossible.
        \end{description}
        This is a contradiction since both cases are impossible.
        So, there is a token on $rr_3$ and, consequently, $U$ has to contain the edge $rr_3$.
        Since there is a token on edge $rr_3$ and each of the spiked graphs $S_i$'s for $i \in [3^{k-1}]$, by the connectivity of $U$, there must be a token on all the edges in the $(k-1, f(k,n))$-claw contained in $C_3$. 
        
        Now, suppose, for a contradiction, that $pw(R_{i}) \geq 2$ for all $i \in [3^{k-1}]$.
        Then, by point 1 of Lemma $\ref{lem:pw_structural}$, we have a ternary tree of depth $k-1$, where each of its leaves are $1$-summed with connected graphs of pathwidth at least $2$, so this means $U$ contains a subgraph with pathwidth at least $k-1+2 = k+1$, contradicting the pathwidth $\leq k$ constraints.

        So, there exists some $i \in [3^{k-1}]$ where $R_i$ is a connected graph with pathwidth $1$ (i.e. a caterpillar).
        Without loss of generality, we assume $R_i = R_1$.
        Recall from the previous observation, there is at least a token on some spike at $s$ for every vertex $s \in \bigcup_{i=1}^{3^{k-1}} V(S_1)$.
        Note that $R_1$ is a subgraph of the $1$-sum of path $P_{2 \cdot 3^{k-1}+1}$ and a $g(k,n)$-spiked graph.
        Furthermore, note that $f(k,n) \geq 2$ since $k \geq 2$, by assumption, and $n \geq 1$ since $G$ is a non-empty connected graph.
        As mentioned, there is a token on all the edges in the $(k-1, f(k,n))$-claw contained in $C_3$, so $R_1$ must contain all the edges of the path $P_{2 \cdot 3^{k-1}+1}$, which has length $f(k,n) \geq 2$.
        
        Since $R_1$ is a connected pathwidth $1$ graph (equivalently, a caterpillar), and the path $P_{2 \cdot 3^{k-1}+1}$ has length $\geq 2$, and there is a token on each $s$-spike for each vertex $s \in V(G_{1})$.
        Then, by Lemma \ref{lem:spike_hamiltonian}, there is a Hamiltonian $v_{G_1}$-path on $G_1$, where $v_{G_1}$ is the vertex we $1$-summed with the tail $P_{2 \cdot 3^{k-1}+1}$.
        Since $G_1$ is just a copy of $G$, there is a Hamiltonian $v$-path for $G$ as desired.
    \end{description}

    To conclude, we show that this is a polynomial-time reduction.
    Let the size of $G$ be $N + M$, where $N = \abs{V(G)}$ and $M = \abs{E(G)}$.
    Since $M < N^2$, it suffices to show that our input $(G_{path}(v), E_s, E_t)$ has size polynomial to $N$.
    By the construction of $G_{path}(v)$, it contains a $(k, f(k,n))$-claw subgraph and $3^{k-1}$ $g(k,N)$-spiked graphs of $G$.
    By Proposition \ref{prop:claw_size}, the $(k, f(k,n))$-claw subgraph has $\frac{3^{k+1}-1}{2} + 3^k \cdot f(k,N) - 1$ edges, thus $\frac{3^{k+1}-1}{2} + 3^k \cdot f(k,N)$ vertices. 
    The spiked graphs have $3^{k-1} \cdot (N + N \cdot g(k,N))$ vertices in total. 
    Recall $g(k,N) = 3^{k-1}N^3$ and $f(k,N) = 3^{k-1} \cdot (N-1 + N \cdot g(k,N)) + 1$. Then,  
    \begin{align*}
        \abs{V(G_{path}(v))} &\leq \frac{3^{k+1}-1}{2} + 3^k \cdot f(k,N) + 3^{k-1} \cdot (N + N \cdot g(k,N)) \\
        &\leq 3^{2k} + 3^{k}(3N \cdot g(k,N)) + 3^{k} \cdot (3N \cdot g(k,N)) \\
        &\leq 3^{2k}(N \cdot g(k,N)) + 3^{2k}(N \cdot g(k,N)) + 3^{2k}(N \cdot g(k,N)) \\
        &\leq 3^{3k} \cdot N^4
    \end{align*}
    Since $k$ is a constant, the size of $G_{path}(v)$ is polynomial in $N$.
    Similarly, note that the size of $E_s, E_t$ are both polynomially bounded above by $\abs{E(G_{path}(v))} \leq \abs{V(G_{path}(v))}^2$.
    Hence, we have a polynomial-time reduction as desired.
\end{proof}

\section{Graphs with Bounded Treewidth and Planar Graphs}\label{app:treewidth}

\subsection{Omitted Proofs in Section \ref{subsec:reconf_treewidth}}\label{app:omitted_treewidth}

\twkstruct*

\begin{proof}
    First, we prove part $1$. Let $T_G, T_H$ be a minimal tree-decomposition for $G$ and $H$, respectively. 
    Without loss of generality, let the $1$-summed vertex be $u$. 
    For each bag in $T_H$ that contains the vertex $v$, rename the vertex $v$ as $u$.
    Let $B_G$ and $B_H$ be bags in $T_G$ and $T_H$ that contain the vertex $u$, respectively.
    Create a new tree $T$ by joining $T_G$ and $T_H$ by an edge through the vertices corresponding to the bags $B_G$ and $B_H$.
    Then, $T$ is a valid tree-decomposition of $G \oplus_{u,v} H$.
    Since the maximum bag size of $T$ is the same as that of $T_H$ or $T_G$, $tw(G \oplus_{u,v} H) \leq k$. 
    
    Now, for part $2$. Let $T$ be a minimal tree-decomposition of $G$.
    Let $G'$ be the resultant graph of the $m$-subdivision and $s_1, \cdots, s_m$ be the subdivided vertices.
    For each bag $B$ containing $uv$, we define a new bag $B_i = \set{u,v,s_i}$ such that $B_i$ is adjacent to $B$ for all $i \in [m]$.
    We call this new tree $T'$.
    Note that $T'$ is a valid tree-decomposition of $G'$ and the maximum bag-size is still $\leq k+1$, since the new bags are only of size $3$, so $tw(G') \leq k$.
\end{proof}

\section{Reconfiguring with Extra Resources}

\subsection{Omitted Proofs in Section \ref{sec:technical3}}\label{app:extra_buffer}

\buffersimple*

\begin{proof}\label{pf:buffer_simple}
    For (1), this is true because an instance of \textsc{tree reconfiguration} is always reconfigurable when $G$ is connected \cite{hanaka_20}, so the minimum required buffer is $0$. 
    
    For (2), cycle reconfiguration is not possible under the usual token-jumping rule \cite{hanaka_20}, and moving any edge into the buffer only results in a path. Thus, the minimum required buffer is infinity.
    
    For (3), we can use a slight variant of \textsc{Path reconfiguration} instance used in \cite{hanaka_20}. 
    We suppose $n \geq 21$, and let $m = \floor{\frac{n-21}{10}}+3$ (so $m \geq 2$).
    Define the input graph to be $G = P_1 \oplus_{p_1, p_2} P_2 \oplus_{p2, p_3} P_3$, where $P_1$, $P_2$, $P_3$ are paths of lengths $m$ respectively, and $p_1$, $p_2$, $p_3$ are one of their endpoints, respectively.
    Let $E_s = E(P_1) \cup E(P_3)$ and $E_t = E(P_2) \cup E(P_3)$, then the instance $(G, E_s, E_t)$ of \textsc{Path reconfiguration} has size $\ell := (3m+1) + 3m + 2m + 2m = 10m + 1$.
    It is easy to see that $(G, E_s, E_t)$ has minimum required buffer $m-1$ since any edge we try to move from $E(P_1)$ to $E(P_2)$ either disconnects the path or makes it into a tree. 
    The size $m-1$ buffer will be used to shorten $E_s$ from a length $2m$ path down to a path with length $m+1$, keeping $P_2$ and one edge in $P_1$. Then we reconfiguration the last remaining edge in $P_1$ to $P_3$, and add all tokens in the buffer to the rest of $P_3$. 
    Now, define a new instance $(G', E_s, E_t)$ where $G'$ is simply $G$ with $n-\ell$ new vertices (not incident to any edge), which makes this an instance of size $n$.
    Note that we can reconfigure $(G', E_s, E_t)$ with $k$ buffer if and only if $(G, E_s, E_t)$ is also reconfigurable with $k$ buffer, since there is no extra added edge in $G'$.
    So, $(G', E_s, E_t)$ is a size-$n$ instance with minimum required buffer $m-1 = \Omega(n)$ as required.
\end{proof}

\bufferpwtw*

\begin{proof}
    For \textsc{connected $k$-bounded pathwidth reconfiguration}, we consider a slight variant of the constructed instance $(G'_{path}(v), E_s, E_t)$ used in the proof of Theorem \ref{thm:pathwidth_k} in Appendix \ref{app:main_result}.

    For $k = 1$, we can use the same proof as the one for the minimum required buffer of \textsc{path reconfiguration} with one single change.
    We will provide an almost verbatim proof here.
    We suppose $n \geq 31$, and let $m = \floor{\frac{n-31}{10}}+3$ (so $m \geq 3$).
    Define the input graph to be $G = P_1 \oplus_{p_1, p_2} P_2 \oplus_{p2, p_3} P_3$, where $P_1$, $P_2$, $P_3$ are paths of lengths $m$ respectively, and $p_1$, $p_2$, $p_3$ are one of their endpoints, respectively.
    Let $E_s = E(P_1) \cup E(P_3)$ and $E_t = E(P_2) \cup E(P_3)$, then the instance $(G, E_s, E_t)$ of \textsc{Path reconfiguration} has size $\ell := (3m+1) + 3m + 2m + 2m = 10m + 1$.
    It is easy to see that $(G, E_s, E_t)$ has minimum required buffer $m-2$ (instead of $m-1$ as in the above proof) since any edge we try to move from $E(P_1)$ to $E(P_2)$ either disconnects the path or makes it into a pathwidth $2$ graph.
    Now, define a new instance $(G', E_s, E_t)$ where $G'$ is simply $G$ with $n-\ell$ new vertices (not incident to any edge), which makes this an instance of size $n$.
    Note that we can reconfigure $(G', E_s, E_t)$ with $k$ buffer if and only if $(G, E_s, E_t)$ is also reconfigurable with $k$ buffer, since there is no extra added edge in $G'$.
    So, $(G', E_s, E_t)$ is a size-$n$ instance with minimum required buffer $m-2 = \Omega(n)$ as required.

    For any fixed $k \geq 2$, we need a slightly different construction, but the structure remains similar.
    We suppose $n \geq 14 \cdot 3^k - 8$, and let $m = \floor{\frac{n- (6 \cdot 3^k -4)}{4 \cdot 3^k -2}}$ (so $m \geq 2$).
    Define the input graph $G$ to be a $(k,m)$-claw and let $E_s$ and $E_t$ be an almost $(k,m)$-claw with the first path in the first and second branch removed, respectively.
    Noting that by Proposition \ref{prop:claw_size}
    \begin{align*}
        \abs{V(G)} &= (3^{k+1}-1)/2 + 3^k m \\
        \abs{E(G)} &= (3^{k+1}-1)/2 + 3^k m - 1 \\
        \abs{E_s} &= (3^{k+1}-1)/2 + 3^k m - 1 - m\\
        \abs{E_t} &= (3^{k+1}-1)/2 + 3^k m - 1 - m, 
    \end{align*}
    then the instance $(G, E_s, E_t)$ of \textsc{Path reconfiguration} has size $\ell := (6 \cdot 3^k -4) + (4 \cdot 3^k -2) \cdot m$. 
    Using the same argument as the proof in \ref{thm:pathwidth_k}, we know $(G, E_s, E_t)$ has minimum required buffer of $m-1$ since any edge we try to move from $E_s \setminus E_t$ to $E_t \setminus E_s$ either disconnects the subgraph or makes it into a pathwidth $k+1$ subgraph.
    Now, define a new instance $(G', E_s, E_t)$ where $G'$ is simply $G$ with $n-\ell$ new vertices (not incident to any edge), which makes this an instance of size $n$.
    Note that we can reconfigure $(G', E_s, E_t)$ with $k$ buffer if and only if $(G', E_s, E_t)$ is also reconfigurable with $k$ buffer, since there is no extra added edge in $G'$.
    So, $(G', E_s, E_t)$ is a size-$n$ instance with minimum required buffer $m-1 = \Omega(n)$ as required.

    Now, for \textsc{$k$-bounded treewidth reconfiguration}, we consider a slight variant of the constructed instance used in the proof of Theorem \ref{thm:minor_closed_hardness} in Appendix \ref{app:treewidth}.
    Let $H$ be a forbidden minor for the class of $k$-bounded treewidth with the minimum number of edges, and let $uv$ and $st$ be two edges of $H$. We know this exists from the proof of \ref{thm:minor_closed_hardness} since it is a \emph{special} minor-closed graph class.
    Note that $n_H = \abs{V(H)}$ and $m_H := \abs{E(H)}$ are constants with respect to the input, since it depends only on the graph class, not the input instance. 
    We suppose $n \geq n_H + 31 m_H - 22$, and let $m = \floor{\frac{n - (n_H + 17 m_H - 11)}{14 m_H - 11}}$ (so $m \geq 1$). 
    Define the instance $(G, E_s, E_t)$ as follows.
    Let $G$ be $H$ by performing a $2(m+1)$-subdivision on all edges except $uv$ and a $(m+1)$-subdivision on $uv$, and let $E_s = E(H) \setminus (S_{uv}^{\uparrow} \cup S_{uv}^{\downarrow})$ and $E_t = E(H) \setminus S_{st}^{\uparrow}$.
    Note that 
    \begin{align*}
        \abs{V(G)} &= n_H + (m_H - 1) \cdot (2m+2) + (m+1) \\
        &= (n_H + 2m_H - 1) + (2 m_H -1) \cdot m \\
        \abs{E(G)} &= m_H + (m_H - 1) \cdot (4m+4) + (2m+2) \\
        &= (5m_H - 2) + (4 m_H -2) \cdot m \\
        \abs{E_s} &= \abs{E(G)} - 2(m+1) \\
        &= (5m_H - 4) + (4 m_H - 4) \cdot m \\
        \abs{E_t} &= \abs{E(G)} - 2(m+1) \\
        &= (5m_H - 4) + (4 m_H - 4) \cdot m,
    \end{align*}
    so the instance $(G, E_s, E_t)$ has size $\ell := (n_H + 17 m_H - 11) + (14 m_H - 11) \cdot m$.
    Using the same argument as the proof in \ref{thm:minor_closed_hardness}, we know $(G, E_s, E_t)$ has minimum required buffer $m$ since by a counting argument, we need at least $m$ tokens to be removed from $G$ or somewhere during the reconfiguration sequence, the token-induced subgraph will contain the forbidden minor $H$ as a minor.
    Now, define a new instance $(G', E_s, E_t)$ where $G'$ is simply $G$ with $n-\ell$ new vertices (not incident to any edge), which makes this an instance of size $n$.
    Note that we can reconfigure $(G', E_s, E_t)$ with $k$ buffer if and only if $(G', E_s, E_t)$ is also reconfigurable with $k$ buffer, since there is no extra added edge in $G'$.
    So, $(G', E_s, E_t)$ is a size-$n$ instance with minimum required buffer $m = \Omega(n)$ as required.

    The $\Omega(n)$ lower bound on the minimum required buffer for \textsc{Planar Graph Reconfiguration} follows from the same argument as \textsc{$k$-bounded treewidth reconfiguration}.
\end{proof}

\buffercactus*

\begin{proof}
    Suppose we have a buffer of size $c_{\max}$, and recall that $\abs{E_s} = \abs{E_t}$, where each of $E_s$ and $E_t$ has at most $c_{\max}$ cycles. 
    Removing an appropriate set of $c_{\max}$ edges from the source yields a forest subgraph $F_s \subseteq E_s$ with $\abs{E_s} - c_{\max}$ edges, and similarly removing $c_{\max}$ edges from the target yields a forest subgraph $F_t \subseteq E_t$ with the same number of edges. 
    Since $F_s$ and $F_t$ are equal-size forests of $G$, hence equal-size independent sets of the graphic matroid, they are reconfigurable into one another by the symmetric exchange property for independent sets.

    We construct the reconfiguration sequence as follows. First, place the $c_{\max}$ tokens of $E_s \setminus F_s$ into the buffer to obtain $F_s$, reconfigure $F_s$ into $F_t$, which we know is always possible by the argument above, and then move the $c_{\max}$ tokens in the buffer onto the missing target edges $E_t \setminus F_t$. 
    This requires an extra buffer of size at most $c_{\max}$ as desired.

    Now suppose both the source $E_s$ and $E_t$ are triangular cacti, that is, all of their components are connected graphs in which every cycle is a triangle and every edge is part of a triangle.
    The core idea of the proof is to view a triangular cactus as a solution to a matroid parity instance, which we refer to \cite{lovasz_86, szigeti_03} for details.
    First, we construct $E_s'$ and $E_t'$ by removing one triangle from $E_s$ and from $E_t$ respectively, and consider the reconfiguration instance $(G, E_s', E_t')$.
    Then, viewing $E_s'$ and $E_t'$ as solutions to a matroid parity instance, we invoke the \textsc{Matroid Parity Reconfiguration} result of \cite{bousquet_23}, which states that any two matroid parity solutions of the same size are reconfigurable into one another provided that a strictly larger parity solution exists. 
    Such a larger solution exists here, since $E_s$ itself is one. 
    This implies there is a reconfiguration sequence from $E_s'$ to $E_t'$ that moves one triangle at a time, where moving a triangle relocates its three tokens to another triangle on the graph $G$ and the token-induced subgraph remains a triangular cacti, and thus cacti.
    We implement such triangle movement by first storing the three tokens of the triangle into the buffer and then placing those three stored tokens at the intended position of the target triangle, which implies the sequence from $E_s'$ to $E_t'$ requires at most $3$ extra buffer space.
    
    Finally, recall that $E_s'$ and $E_t'$ were obtained by removing one triangle from $E_s$ and $E_t$ respectively. 
    We store the three tokens of the triangle removed from $E_s$ in the buffer, so the buffer starts with three tokens, and then we follow the reconfiguration sequence from $E_s'$ to $E_t'$ as above, and finally we place the three initially stored tokens onto the triangle removed from $E_t$. 
    Since the first triangle's three tokens remain in the buffer throughout the reconfiguration from $E_s'$ to $E_t'$, which itself uses at most $3$ extra buffer space, this means at most $6$ tokens can ever be in the buffer simultaneously. 
    Hence, when the source and target are triangular cacti, the minimum required buffer for any instance of \textsc{cacti reconfiguration} is at most $6$, as required.
\end{proof}

\end{document}